\documentclass[12pt]{article}
\usepackage{amsfonts, amsmath, amssymb, mathtools}
\usepackage[unicode=true,pdfusetitle, bookmarks=true,bookmarksnumbered=false,bookmarksopen=false,
 breaklinks=false,pdfborder={0 0 0},pdfborderstyle={},backref=false,colorlinks=true]{hyperref}
\usepackage{setspace}
\usepackage{authblk}
\usepackage{geometry}
\usepackage{natbib}
\usepackage{comment}
\usepackage[ruled,vlined]{algorithm2e}
\usepackage{multirow}
\usepackage{booktabs,array,tabularx}
\usepackage{caption}
\usepackage{graphicx,float}
\usepackage{subcaption}
\usepackage{threeparttable}
\usepackage{longtable}

\setcitestyle{authoryear,open={(},close={)}}
\hypersetup{citecolor=blue, linkcolor=blue, urlcolor=blue}

\newtheorem{theorem}{Theorem}

\newtheorem{assumption}{Assumption}

\newtheorem{lemma}{Lemma}

\newenvironment{proof}[1][Proof]{\noindent\textbf{#1.} }{\ \rule{0.5em}{0.5em}}

\newcommand{\R}{ \mathbb R}
\newcommand{\Ebb}{\mathbb{E}}
\newcommand{\Links}{\mathcal{L}}

\newcommand{\norm}[1]{\left\lVert#1\right\rVert}

\newcommand{\parrow}{\overset{p}{\rightarrow}}

\newcommand{\argmin}{\operatornamewithlimits{argmin}}
\newcommand{\Darrow}{\overset{D}{\rightarrow}}

\begin{document}

\title{Estimating the perturbed utility route choice model with trip-level data}
\author[1,2]{Mogens Fosgerau}
\author[1]{Nikolaj Nielsen}
\author[2]{Thomas Rasmussen}
\author[3]{Rui Yao\footnote{Corresponding author: \href{mailto:rui.yao@technion.ac.il}{rui.yao@technion.ac.il}.}}

\affil[1]{University of Copenhagen}
\affil[2]{Technical University of Denmark}
\affil[3]{Technion}
\maketitle

\begin{abstract}
We provide an estimator for the perturbed utility route choice (PURC) model that works with data at the level of individual trips. The estimator is a nested fixed-point algorithm that combines an upper bias-corrected linear regression problem with a lower individual-level perturbed utility maximization problem. We establish the statistical properties of the microPURC estimator and confirm these results with an experiment using simulated data. Finally, we demonstrate the estimator in practice using a large real-world dataset.
\end{abstract}

\section{Introduction}

This paper proposes an estimator for the perturbed utility route choice (PURC) model \citep{fosgerau_perturbed_2022, fosgerau_bikeability_2023} that applies to trip-level data without aggregation. This makes it possible to estimate the parameters of a PURC model directly from a dataset of observed route choices through a network together with a matrix of network link characteristics, possibly interacted with individual-specific variables. We refer to the new proposed estimator as a microPURC estimator.

This is an improvement over the estimator in \citet{fosgerau_perturbed_2022}, which requires route choice data to be averaged at the origin-destination pair level into origin-destination flows. This aggregation has two costs. First, to avoid noisy averages, trips in origin-destination pairs with few observations must be discarded, resulting in data loss. Second, aggregation discards individual-level information, precluding the use of traveler- or trip-specific characteristics in the cost function. The \citet{fosgerau_perturbed_2022} estimator therefore entails both a loss of data and a loss of information, both of which are avoided by our proposed estimator.

To explain the contribution of the current paper in more detail, we first outline the model specification. The PURC model views the traveler as choosing a flow vector $x$ of non-negative link flows that minimizes a convex function of the form  $x^\top c + F(x)$. Here,  $c$ is a vector of non-negative link costs and $F$ is a convex function that incorporates the network structure. The flow $x$ is required to satisfy flow conservation from the origin to the destination. The observed route $y$ is coded as a vector of zeroes and ones, indicating which links were used. The observed route is assumed to be a random draw from the optimal flow vector, i.e., it is assumed to satisfy  $\Ebb [y] = x$.

If one is willing to aggregate data to the level of origin-destination (OD) flows, the PURC model is also very fast and easy to estimate using the method of moments \citep{fosgerau_perturbed_2022, fosgerau_bikeability_2023}. This approach relies on the first-order condition for the traveler's problem, which includes the gradient $c+\nabla F(x)$ as well as (very many) Lagrange multipliers for the flow conservation constraint. As we will show, the Lagrange multipliers are eliminated by a certain projection matrix $P$, giving us the projected first-order condition
\begin{align}\label{projFOC}
P(c+\nabla F(x))=0.
\end{align}
Writing the link cost vector as the product of a link-characteristics matrix $Z$ and a parameter vector $\beta$, we obtain $P\nabla F(x)=-PZ\beta$, which has the form of a linear regression equation. With many route choice observations for some given OD-pair, we can compute an approximation of the flow $\hat x$ as the average of the observed routes and plug that into the regression equation. With many such OD-pairs, we can estimate $\beta$ by regression. As shown in \citet{fosgerau_perturbed_2022} using simulated data, this approach allows us to recover the true parameters with realistic sample sizes.

However, with real data, aggregating to the common ODs generally entails data loss due to trip trimming and discarding observations in OD pairs with few observations. Aggregation also limits the possibilities for incorporating individual-level information in the model.

Hence, we seek a way to estimate the PURC model with trip-level data. The challenge we face is that we only observe the chosen routes $y$, not the flows $x$ from which the routes are drawn. We can therefore not directly apply the projected first-order condition.

We propose to meet this challenge as follows. Given an initial parameter estimate, we solve the perturbed utility maximization problem to predict the flow $\hat x$ corresponding to each observed chosen route. Next, for each observation, we add a term to the projected first-order condition \eqref{projFOC} that depends on the observed route and that allows $\beta$ to be estimated without bias from the amended first-order condition. This gives us an updated parameter estimate. Generally, as we will show, the updated parameter estimate will improve on the previous estimate.

This procedure can be repeated. On convergence, we show that it yields a consistent $\sqrt{N}$ asymptotically normal estimate of the model parameters. This is our proposed microPURC estimator.
An open-source implementation of the microPURC estimator is available at \url{https://github.com/andyYaoR/micropurc}.

\subsection{Literature review}\label{sec:litreview}

\paragraph{Path-based route choice models} The extensive literature on route choice is reviewed in \citet{prato_route_2009}. Most papers treat route choice as a discrete choice among routes, which poses the challenge that the number of feasible routes in a large network is astronomical. Research has therefore focused on constructing choice sets with good coverage to avoid bias. Another focus has been on developing models with realistic substitution patterns while maintaining tractability.

Estimation of these models is typically performed using maximum likelihood estimation. However, this is not directly available in the PURC context, since the PURC model (by design) allows corner solutions. Hence, the PURC model may predict zero flow on links used by an observed route, which will cause the likelihood to be zero.

\paragraph{Recursive models} Instead of viewing route choice as a discrete choice among a set of routes, a newer line of research considers recursive models, where travelers choose paths link by link in a Markovian fashion. Early work addressed the assignment problem \citep{Dial1971, bell_alternatives_1995, shen_cyclic_1996, baillon_markovian_2008}. Subsequent studies introduced the recursive logit model and its generalizations based on the multivariate extreme value distribution \citep{fosgerau_link_2013, Mai2015, Mai2015a, Mai2016}. The duality properties of these models have been exploited by \citet{oyama_markovian_2022}. However, applying recursive logit to large networks remains challenging due to the need to invert large matrices. \citet{zimmermann_tutorial_2020} provides a comprehensive tutorial.

\noindent\textbf{The PURC model.}  \citet{fosgerau_perturbed_2022, fosgerau_bikeability_2023} depart from the discrete choice paradigm by modeling travelers as maximizing a perturbed utility function \citep{McFadden2012}. It requires no choice set and uses the entire network directly. Substitution patterns arise from the network structure via a perturbation function defined as a sum of convex link-flow functions.  \citet{yao_perturbed_2024} propose a fast algorithm for computing equilibrium assignment in large networks, and \citet{fosgerau_sensitivity_2025} consider sensitivity analysis of the PURC equilibrium assignment.

With aggregate route choice data, PURC can be estimated using linear regression \citep{fosgerau_perturbed_2022}, similarly to the estimation of the inverse product differentiation logit model from market shares \citep{fosgerau_inverse_2024}. In this paper, we extend the regression equation by adding a debiasing term to correct for the bias that results from replacing observed market shares with noisy estimates \citep{chernozhukov2018double,chernozhukov2020locally}.

\paragraph{Nested Fixed-Point Estimation} The proposed estimator belongs to the class of nested fixed-point algorithms, which are widely used in structural econometrics to estimate models with an inner optimization or equilibrium condition. Introduced by \citet{rust1987optimal} for dynamic discrete choice models, nested fixed-point algorithms combine an outer optimization over parameters with an inner routine that solves the model for optimality. In our context, the inner loop solves the perturbed utility maximization problem for each trip, while the outer loop adjusts parameters based on the bias-corrected first-order condition. Our approach has similarities to the constrained optimization methods used in demand estimation \citep{berry1995automobile, nevo_practitioners_2000}. The convergence and asymptotic properties of our estimator build on the theory of iterative fixed-point methods in econometrics \citep{aguirregabiria2007sequential, su2012constrained}.

\subsection{Layout}

The paper proceeds as follows: Section \ref{sec:setup} sets up the model. Section \ref{sec:estimation} sets up the estimator, discusses its implementation, and derives its asymptotic properties. Section \ref{sec:sim} presents simulation results that illustrate and verify the theoretical predictions. Section \ref{sec:scalability} examines the computational scalability of the estimator. Section \ref{sec:empirical_application} demonstrates the feasibility of the estimator with a large-scale network. Section \ref{sec:conclusion} concludes. The appendix contains an index of notation, supporting mathematical results, proofs of the results in the main text, and a discussion of the computation of optimal flows.

\section{Model setup}\label{sec:setup}

We model the choice of route through a network for travelers who make a single trip from an origin to a destination. The network is connected with nodes and links $\left( \mathcal{N},\mathcal{L}\right) $. Links are indexed by $ij$, where $i$ and $j$ are the start and end nodes. Nodes are also indexed by  $v$ depending on the context. The network structure is specified by a node-link incidence matrix $A=\left\{ a_{v,ij}|v\in \mathcal{N},ij\in \mathcal{L}\right\} $ with entries
\begin{equation*}
a_{v,ij}=\left\{
\begin{array}{rc}
-1, & v=i \\
1, & v=j \\
0, & \text{otherwise.}%
\end{array}%
\right.
\end{equation*}%
We consider here a generic random trip made by a traveler and omit notation that refers to the traveler. For a flow vector $x\in \mathbb{R}_{+}^{\mathcal{L}}$, flow conservation for a trip is expressed as $Ax=b$, where $b\in \mathbb{R}^{\mathcal{N}}$ has components that are zero except $b_{v}=-1$ when $v$ is the origin of the trip and $b_{v}=1$ when $v$ is the destination of the trip.

We observe $\left( y, Z, b\right) $, where $y\in \{0,1\}^{
\mathcal{L}}$ is a vector indicating the links used, and $Z$ is a matrix with a row for each network link and columns for network characteristics, such that $c = Z\beta $ is a link cost vector.
The matrix $Z$ comprises network characteristics but can also include
interactions of traveler or trip-specific information with link
characteristics. The link cost vector is assumed to be positive at the true value of $\beta $, i.e., $Z \beta_{0} \gg 0$.

Given an origin-destination pair, the traveler chooses a flow $x$ that minimizes a perturbed cost function. For a flow vector $x\in \mathbb{R}_+^{\mathcal{L}}$, we define the perturbation function as the sum of link perturbation functions
\begin{align}  \label{eq:sum_link_perturbations}
F(x) = \sum_{ij\in \mathcal{L}} F_{ij}(x_{ij}).
\end{align}
This specification generalizes the specification in  \citet{fosgerau_perturbed_2022}, who define link perturbation functions by multiplying a generic perturbation function by link lengths to ensure that the overall perturbation function $F$ is invariant with respect to link splitting. The present formulation allows link perturbation functions to be specific to each link.

Each link perturbation function $F_{ij}:  \mathbb{R}_{+}\rightarrow
\mathbb{R}_{+}$ is assumed to be continuously differentiable, strictly convex with positive second derivative, and strictly increasing, with $F_{ij}\left( 0\right) =F^{\prime
}_{ij}\left(0\right) =0$ and range equal to $ \mathbb{R}_{+}$.\footnote{Examples of such functions include the entropic $F_{ij}(x_{ij}) = l_{ij}[(1+x_{ij}) \ln (1+ x_{ij}) - x_{ij}]$ and the quadratic $F_{ij}(x_{ij}) = l_{ij} (x_{ij})^2$, where the length of link $ij$, $l_{ij}$, is included to ensure that the perturbation is invariant with respect to link splitting. See \citet{fosgerau_perturbed_2022}.\label{fn:examples}}

The optimal flow vector for the traveler, given cost vector $c$, is the solution to the convex minimization problem

\begin{equation*}
\label{eq:x_hat} \hat x(c,b) =\operatornamewithlimits{argmin}_{x\in  \mathbb{%
R}_+^{ \mathcal{L}}}\left\{ c^\top x+ F\left( x\right) | Ax=b\right\}.
\end{equation*}
When link costs are positive, the set of links carrying positive optimal flow forms a directed acyclic subnetwork \citep{fosgerau_perturbed_2022}.

Observed individual trips $y$ are random route realizations. We assume that their conditional expectation, given vector $b$ and attributes $Z$, coincides with the PURC optimal flow evaluated
at the true parameter $\beta_0$, i.e.,
\begin{equation}\label{eq:Ey}
\Ebb[y|Z,b]=\hat x\left( Z \beta_{0},b \right) .
\end{equation}
This conditional moment condition connects trip-level route observations with the PURC model. The condition may be interpreted as the outcome of a Markovian route choice process, as shown in the following:

Since $x$ solves a cost-minimization problem with strictly positive link costs, it contains no circulation, since removing any circulation would strictly lower cost. As $o$ is the unique source, this also rules out incoming links at $o$ and, by the same reasoning, outgoing links at $d$. It follows that $\text{outflow}(o)= \sum_{(o,j) \in \mathcal{L}} x_{oj}= 1$. Now order the nodes in the active subnetwork topologically as $o=v_1,\dots,v_M=d$. Consider a random walk that starts at $o$ and, at any node $v$, moves to node $j$ along link $(v,j)$ with conditional probability
\[
p_{vj}=\frac{x_{vj}}{\text{outflow}(v)},
\]
until it reaches $d$. Let $q_{ij}$ denote the probability that the walk traverses link $(i,j)$.

We show by induction on $m=1,\dots,M-1$ that the probability of visiting node $v_m$ is $\text{outflow}(v_m)$. For $m=1$, this is immediate since $\Pr[\text{visit }o]=1=\text{outflow}(o)$. Now fix $m>1$ and suppose that $q_{v_k,v_m}=x_{v_k,v_m}$ for every $k<m$. Then
\begin{equation}\label{eq:markov_sampling}
    \Pr[\text{visit }v_m]=\sum_{k:(v_k,v_m)\in\mathcal L} q_{v_k,v_m}
=\sum_{k:(v_k,v_m)\in\mathcal L} x_{v_k,v_m}=\text{inflow}(v_m)=\text{outflow}(v_m),
\end{equation}
where the last equality uses flow conservation at the traversing node $v_m$, for which $b_{v_m}=0$. Hence, for every active link $(v_m,j)$,
\[
q_{v_m,j}=\Pr[\text{visit }v_m]\,p_{v_m j}=x_{v_m,j},
\]
which completes the induction. Since $d$ has no outgoing links, this establishes $q_{ij}=x_{ij}$ for every active link. Therefore $\Ebb[y_{ij}]=q_{ij}=x_{ij}$ for all $ij$, so
\[
\Ebb[y\mid Z,b]=x=\hat x(Z\beta_0,b),
\]
which confirms \eqref{eq:Ey}. Thus, conditional on $(Z,b)$, the traveler can be viewed as selecting an outgoing link at each node with probability proportional to the optimal PURC flow until reaching $d$. We use this Markov process to simulate routes in Section~\ref{sec:sim}.

\subsection{Optimizing behavior}\label{sec:optimizing_behavior}

To derive an estimator for $\beta_{0}$, we begin by analyzing the generic traveler's cost minimization problem. Given cost vector $c=Z\beta$, the Lagrangian for the perturbed cost minimization problem is
\begin{equation}  \label{eq_Lagrange}
\Lambda( x,\eta) =x ^\top c + F(x) +\eta ^\top \left( Ax-b\right) ,\; x\in
\mathbb{R}_+^{\mathcal{L}}.
\end{equation}

Let $\hat{B}=\hat{B}(c,b)=\mathrm{diag}\left(1_{\hat{x}>0}\right)$ be the matrix with ones on the diagonal corresponding to interior solutions for the link flows $\hat{x}_{ij}$. The first-order condition for the cost minimization problem for the positive entries of $\hat{x}$ can then be written%
\begin{equation}
\hat{B}\left( c+\nabla F\left( \hat{x}\right) +A^{\top }\hat{\eta}\right) =0,
\label{eq:1}
\end{equation}%
where $\nabla F(\hat{x})$, the gradient of $F$, has components $F_{ij}^{\prime
}(\hat{x}_{ij})$.

Eq. \eqref{eq:1} shows that, on the active subnetwork, the marginal cost $c+\nabla F(\hat x)$ is given by $-A^\top\hat\eta$, restricted to the active links. 
We eliminate $\hat\eta$ by projecting Eq. \eqref{eq:1} onto a subspace orthogonal to the multiplier term. Let $\hat P$ denote the orthogonal projection onto the linear subspace
\begin{eqnarray}\label{eq:subspace}
\{ x\in \mathbb R^{\mathcal{L}}:A x =0,  \hat Bx =x\},
\end{eqnarray}
i.e., the intersection of the null space of $A$ with the subspace where components corresponding to inactive links are zero. This subspace is exactly the space of \emph{circulations} on the active subnetwork: perturbations of the flow $\hat x$ that reroute traffic among alternative active paths or cycles between origin and destination while leaving net inflow and outflow unchanged at every node. Because such perturbations lie in the null space of $A$, and the null space of $A$ is orthogonal to its row space $\mathrm{range}(A^\top)$, they are automatically orthogonal to any vector of potential differences $A^\top\eta$, hence $\hat P A^\top = 0$. In other words, premultiplying by $\hat P$ eliminates the node potentials,
whatever their value. (Lemma \ref{lem:P_hat} in Appendix \ref{app:Phat} gives the explicit formula for $\hat P$ and verifies this and its other properties formally.) Premultiplying Eq. \eqref{eq:1} by $\hat{P}$ therefore yields the projected first-order condition, which no longer
involves $\hat\eta$ and is the starting point for our estimation strategy:
\begin{equation}\label{eq:projectedFOC}
\hat{P}(c+\nabla F(\hat{x}))=0.
\end{equation}
Intuitively, \eqref{eq:projectedFOC} says that marginal cost, measured along any feasible reallocation of flow within the active subnetwork, cannot be reduced.

We can think of $\hat P$ as a matrix of instruments: its rows are orthogonal to the marginal cost by construction, regardless of the
unobserved multipliers $\hat\eta$. The estimator we present is essentially an instrumental-variables estimator, with the caveat that $\hat P$ depends on $\hat x$.

\paragraph{Active link detection}
Constructing $\hat P$ requires the active set $\{ij:\hat x_{ij}>0\}$ encoded in $\hat B$, which we read from the dual slack. Define the dual slack of link $ij\in\mathcal{L}$ as
\begin{equation}\label{eq:active_threshold}
    s_{ij}\coloneqq c_{ij}+(A^\top\hat\eta)_{ij},
\end{equation}
and note that link $ij$ is active if and only if $s_{ij}<0$. Stationarity gives $c+\nabla F(\hat x)+A^\top\hat\eta=\hat\mu$, where $\hat\mu\ge 0$ is the multiplier of the non-negativity constraint $x\ge 0$, so $s_{ij}=\hat\mu_{ij}-F_{ij}'(\hat x_{ij})$. Since $F_{ij}'(0)=0$, complementarity splits this into $s_{ij}=-F_{ij}'(\hat x_{ij})<0$ on active links and $s_{ij}=\hat\mu_{ij}\ge 0$ on inactive ones, so $s_{ij}$ is the signed margin to the switching boundary, in cost units. 
Because it is built from the costs and node potentials, $s_{ij}$ is more robust to floating errors in flow: an inactive link keeps a strictly positive slack $\hat\mu_{ij}$ even when $\hat x_{ij}$ contains a small positive residual, so $s_{ij} \geq 0$ correctly classifies inactive links, and similarly for classifying active links.

\section{Estimation}\label{sec:estimation}
We now turn to estimation of the cost parameters in the link cost function $c = Z \beta$. We impose the following regularity conditions:

\begin{assumption}
\; \label{as:regularity}

\begin{itemize}
\item[a)] $(y^n,Z^n,b^n)$ are i.i.d. realizations from the joint distribution of a random variable $(y,Z,b)$.

\item[b)] There exists a unique $\beta_0\in \mathbb{R}^K$ such that for
all $ Z, b$, $\Ebb\left[ y|Z,b\right]=\hat
x( Z\beta_0,b)$ and $Z\beta_0 \gg 0$. There exists a compact set $\Theta$ such that $\beta_0\in \mathrm{int} (\Theta)$.

\item[c)] $F:\R^{\mathcal{L}}_+\rightarrow \R_+$ is four times continuously differentiable.

\item[d)] $\Ebb\left[Z^\top \hat P Z \right]$ has full rank, where $\hat P=\hat P(\hat x(Z\beta_0,b ))$ and $\Ebb\left[\left\| Z\right\|^2\right]<\infty$

\item[e)] There exists an open neighborhood $\mathcal{B}_0$ around $\beta_0$ such that $1_{\hat x(Z\beta,b)>0}=1_{\hat x(Z\beta_0,b)>0}$ for almost all $Z,b$ and all $\beta\in \mathcal{B}_0$.
\end{itemize}
\end{assumption}

Items a) and b) are standard regularity conditions: trip-level observations are i.i.d., and a unique true parameter vector $\beta_0$ rationalizes the observations such that the conditional mean of the observed route coincides with the PURC optimal flow. Assumption b) is a conditional moment condition, $\Ebb[y\mid Z,b]=\hat x(Z\beta_0,b)$: it requires only that the expected route choice matches the PURC-optimal flow, not that route choices follow any specific probability law over paths. This is weaker than specifying path probabilities as done in path-based models and the recursive logit model.
Item c) ensures sufficient smoothness of the perturbation function $F$ for the asymptotic results -  this assumption is satisfied for the perturbation functions in footnote \ref{fn:examples}.
The assumption of main practical relevance is item d), which requires that the projected regressors contain sufficient variation to identify $\beta_0$ and yield a nonsingular asymptotic covariance matrix.
Item e) is a local constant active–set condition to ensure local differentiability of the PURC solution with respect to $\beta_0$. It is a mild assumption, since, as shown in \citet{fosgerau_sensitivity_2025}, the set of links with positive flow is constant in an open neighborhood of any point in the space of link cost vectors, except on a closed set of Lebesgue measure zero.

\paragraph{Intuition for identification of $\beta_0$}  Assumption~\ref{as:regularity}(b) postulates that a unique parameter value rationalizes the observed conditional mean route choice across the population of trips; by  Lemma~\ref{lem:unique_RV} below, this is equivalent to $\beta_0$ uniquely minimizing the residual variance between observed and predicted flow. Locally, identification is driven by the projected first-order condition \eqref{eq:projectedFOC}, which requires the marginal cost to be equalized across all alternative routes within the active subnetwork for a given trip. Whenever a traveler's active subnetwork contains more than one route between origin and destination, comparing the link attributes along these alternative routes reveals the rate at which travelers trade off one attribute against another. Assumption~\ref{as:regularity}(d) requires that this cross-route variation in attributes be rich enough to separately identify each coefficient. Identification can therefore fail in two ways: if $Z$ varies collinearly across the alternative routes available in the data (e.g., if travel time is everywhere proportional to distance), or if trips in the sample have essentially just one active route. Nevertheless, with sufficiently rich data, identification is expected in practice. 

Consider the residual variance function
\begin{eqnarray}\label{eq:RV}
        RV(\beta)=\Ebb\left[\norm{y-\hat x(Z\beta,b)}^2 \right],
\end{eqnarray}
which measures the expected squared distance between the observed route choice $y$ and the predicted flow $\hat x(Z\beta,b)$.

Lemma \ref{lem:unique_RV}, proved in the Appendix, shows that the true parameter $\beta_0$ is the unique minimizer of the residual variance function.
\begin{lemma}\label{lem:unique_RV}
  $\beta_0$ is the unique minimizer of $RV$.
\end{lemma}
However, directly minimizing the residual variance is computationally difficult, since  $\hat{x}(Z\beta, b)$ is itself the solution to a constrained convex program (the PURC model) that allows corner solutions. Direct optimization using gradient-based methods requires the sensitivity of the solution with respect to parameters to be computed many times, which can be challenging~\citep{fosgerau_sensitivity_2025}.
We therefore seek an approach that exploits the PURC problem structure to make estimation easier and, not least, much faster.

Based on the projected first-order condition \eqref{eq:projectedFOC}, we now construct a map that has the true parameter as a fixed point. To do so, we first introduce a function $T$ as a first-order approximation to the projected first-order condition at the point $x$ in the direction of the choice variable $y$. For given $\left( \beta ,x,Z,y\right)$, define
\begin{equation}\label{eq:T_function}
    T\left( \beta ,x,Z,y\right) =\hat{P}(x) \left( Z\beta +\nabla F\left( x\right)
    +\nabla ^{2}F\left( x\right) \cdot \left( y-x\right) \right),
    \end{equation}%
where $\nabla ^{2}F\left( x\right) \cdot \left( y-x\right)$ is the first-order approximation of $\nabla F\left( x\right)$ around $x$, evaluated at the point $y$.

To unpack the meaning of this function, write $x_0=\hat x(Z\beta_0,b)$ for the optimal flow at the true parameter. Consider the term $\nabla F(x) + \nabla^2 F(x)(y-x)$, where $x=\hat{x}(Z \beta,b)$ is the predicted flow at the current value of $\beta$. At the true parameter, where $x=x_0$, taking the conditional expectation given $(Z,b)$ gives
\begin{equation*}
\Ebb\bigl[\nabla F(x_0) + \nabla^2 F(x_0)(y-x_0)\,\big|\,Z,b\bigr] = \nabla F(x_0),
\end{equation*}
since the second term vanishes in conditional expectation because $\Ebb[y\mid Z,b]=x_0$ by Eq.~\eqref{eq:Ey}. Because $x_0$, and hence $\hat P(x_0)$, are fixed given $(Z,b)$, it follows that
\begin{equation*}
\Ebb\bigl[T(\beta_0,x_0,Z,y)\,\big|\,Z,b\bigr] = \hat P(x_0)\bigl(Z\beta_0+\nabla F(x_0)\bigr)=0,
\end{equation*}
where the last equality is the projected first-order condition \eqref{eq:projectedFOC} evaluated at $\beta_0$. 

Based on this approximation, we next define an iterative map $\beta\rightarrow \Phi(\beta)$ by
\begin{eqnarray}\label{eq:Phi}
    \Phi(\beta) \in \argmin_{\beta'} \Ebb \left[ \norm{T(\beta',\hat x(Z\beta,b),Z,y)}^2\right].
\end{eqnarray}
The mechanics of this map are the following. For a given value of the parameter $\beta$, we compute the corresponding predicted flow $\hat x(Z\beta,b)$, then $\Phi(\beta)$ as the value of $ \beta'$ that minimizes the norm of the approximate projected first-order condition $T(\beta',\hat x(Z\beta,b),Z,y)$.
The value $\Phi(\beta)$ is uniquely determined for $\beta$ in a neighborhood of the true parameter $\beta_0$ by the regularity Assumption \ref{as:regularity}; see Lemma \ref{lem:Phi_properties} in the Appendix.

The function $T( \beta',\hat x(Z\beta,b),Z,y)$ is linear as a function of the $ \beta'$, which makes the iterative map available in closed-form as
\begin{eqnarray}\label{eq:Phi_closedform}
    \Phi(\beta) = -\left( \Ebb\left[ Z^{\top }\hat{P} Z\right]
\right) ^{-1}\Ebb\left[ Z^{\top }\hat {P}\left( \nabla F(\hat x)+\nabla ^{2}F(\hat x )(y-\hat x)\right) \right],
\end{eqnarray}
where $\hat x = \hat x(Z \beta,b)$ is the optimal flow corresponding to $\beta$, the projection $\hat{P}$ is constructed based on $\hat x$, and $\Ebb [Z^\top  \hat P Z]$ has been assumed to be invertible.

In the Appendix, we prove the following result.

\begin{lemma}\label{lem:fxp_foc}
    Any point $\beta^*$ is a fixed point of $\Phi$ if and only if it satisfies the first-order condition  for problem \eqref{eq:Phi}
    \begin{eqnarray}
        \Ebb \left[ Z^\top P^*\nabla^2 F(x^*)(\hat x(Z\beta_0,b)-x^*)\right]=0,
    \end{eqnarray}
    where $x^*=\hat x(Z\beta^*,b)$ and $P^*=\hat P(x^*)$.
\end{lemma}

Lemma \ref{lem:fxp_foc} shows in particular that the true parameter $\beta_0$ is a fixed point of $\Phi$. In Lemma \ref{lem:unique_RV} we showed that the true parameter is the unique minimizer of the residual variance. Hence, the true parameter also minimizes the residual variance over the set of fixed points of $\Phi$.

\begin{theorem}\label{thm:estimator} The true parameter $\beta_0$ minimizes the residual variance over the set of fixed points of $\Phi$, i.e.
\begin{eqnarray}
    \beta_0=\argmin_{\beta} RV(\beta) \quad \textrm{s.t.} \quad \Phi(\beta)=\beta.
\end{eqnarray}
\end{theorem}

In addition, $\Phi$ is a local contraction, with quadratic convergence towards $\beta_0$, as shown by the following Lemma, proved in the Appendix.
\begin{lemma}\label{lem:quadratic_convergence}
    Given Assumption \ref{as:regularity}, there exists a finite constant $C$  such that $\norm{\Phi(\beta)-\beta_0}\leq C\norm{\beta-\beta_0}^2$, for $\beta$ sufficiently close to $\beta_0$.
\end{lemma}
For $\beta$ sufficiently close to $\beta_0$, we have $C\norm{\beta-\beta_0}<1$. Then the Lemma guarantees that $\norm{  \Phi(\beta)-\beta_0}$ is strictly smaller than $\norm{\beta-\beta_0}$. Repeated application of $\Phi$ then results in a sequence which quickly converges to $\beta_0$. There is no similar guarantee for the potential other fixed points of $\Phi$, which means the iterative map may not converge to such points.

\paragraph{Quadratic perturbation}

The iterative map $\Phi$ has a particularly elegant form when $F$ is a
quadratic function. We would favor using
\begin{equation}\label{eq:quadF}
    F(x)=\frac{1}{2}x^\top L x,
\end{equation} where $L$ is a diagonal matrix of positive link lengths, since this makes the perturbation function invariant with respect to link splitting. In that case,
\begin{eqnarray*}
\nabla F(x) + \nabla^2 F(x) (y-x)=Lx+ L(y-x)= Ly,
\end{eqnarray*}
such that the iterative map reduces to
\begin{eqnarray*}
\Phi( \beta) = -\Ebb \left[ Z^{\top}\hat P Z \right]^{-1} \Ebb \left[ Z^{\top}\hat P L y \right].
\end{eqnarray*}
In this case, the flow $\hat x$ enters only with the set of active flows into the projection matrix $\hat  P$.


\subsection{Implementation and convergence}

A sample of trips has observations $(y^{n},Z^{n},b^{n}), n=1,\ldots ,N$ that are assumed to be i.i.d. realizations of the random variables $(y,Z,b)$. We suggest estimating $\beta_0$ by the sample analog of the procedure described above, i.e.
\begin{eqnarray}\label{eq:tilde_beta}
    \tilde \beta \equiv \argmin_{\beta} RV_N(\beta) \quad \textrm{s.t.} \quad \Phi_N(\beta)=\beta,
\end{eqnarray}
where $RV_N$ and $\Phi_N$ replace population expectations with sample averages in \eqref{eq:RV} and \eqref{eq:Phi}.

Lemma~\ref{lem:quadratic_convergence} shows that $\Phi(\cdot)$ is a local contraction near $\beta_0$. In general, however, away from $\beta_0$, iterating  $\tilde \beta_{(t+1)} = \Phi_N (\tilde \beta_{(t)})$ can be numerically unstable. We therefore use a damped fixed-point scheme of Krasnoselskii-Mann (KM) type~\citep{mann1953mean} with a non-monotone line search~\citep{grippo1986nonmonotone}. We summarize the proposed approach in the following Algorithm~\ref{iterated_algo}.

\begin{algorithm}[H]\label{iterated_algo}
  \caption{Iterated algorithm}

\begin{itemize}
    \item[Step 0: ]Choose starting value $\tilde \beta_{(0)}$, a convergence tolerance $\tau>0$, a reduction factor $\gamma\in(0,1)$, and integers $M \ge 0, m_{\max} \ge 0$.

    \item[Step 1: ]Given $\tilde \beta_{(t)}$, compute $\hat x^n_{(t)}=\hat x(Z^n \tilde \beta_{(t)},b^n)$, and $\hat P^n_{(t)}=\hat P(\hat x^n_{(t)})$ for $n =1, ..., N$.
\item[ Step 2: ]Compute step $\tilde \beta_{(t+0.5)}=\Phi_N(\tilde \beta_{(t)})$, i.e.
\begin{eqnarray*}
\tilde \beta_{(t+0.5)}=-\left(\Ebb_N\left[Z^{n\top}\hat P^n_{(t)}Z^n\right]%
\right)^{-1}\Ebb_N\left[Z^{n\top}\hat P^n_{(t)} \left( \nabla F(\hat
x^n_{(t)}) + \nabla^2 F(\hat x^n_{(t)}) (y^n-\hat x^n_{(t)} )\right) \right],
\end{eqnarray*}
where $\Ebb_N[\cdot]$ denotes the empirical average over $n$.

\item[Step 3: ] KM step: Update $\tilde \beta_{(t+1)} = (1 - \alpha_{(t)}) \tilde \beta_{(t)} + \alpha_{(t)} \tilde \beta_{(t+0.5)}$ with line search (elaborated below) to determine step size $\alpha_{(t)}$.

\item[Step 4: ] If $\norm{\tilde \beta_{(t+1)} - \Phi_N(\tilde \beta
_{(t+1)})} > \tau$, repeat steps 1-3.

\end{itemize}

\end{algorithm}
\bigskip

Elaborating on the KM iteration (Step 3), we propose the following non-monotone line search scheme to determine the step size $\alpha_{(t)}$. Define the merit function
\[
    Q_N(\beta) \coloneqq \frac{1}{2}\|\beta-\Phi_N(\beta)\|^2.
\]
Letting $m$ run from $0$ to $m_{\max}$, set $\alpha_{(t)} = \gamma^m$ and accept the first $\alpha_{(t)}$ satisfying the non-monotone condition
\[
    Q_N(\tilde\beta_{(t+1)})
    \;\le\;
    \max_{0\le i\le \min\{t, M\}} Q_N(\tilde\beta_{(t-i)}).
  \]
Otherwise, let $\alpha_{(t)}=\gamma^{m_{\max}}$. We discuss in Appendix~\ref{app:Phat} and~\ref{app:implementation_details} how to compute the projection and $\hat x^n_{(t)}$ efficiently.

\paragraph{Practical guidance} Note that we do not guarantee convergence to the true $\beta_0$. This is a general issue with nested fixed-point estimators of this kind; see, for example, \citet{Knittel2014}. Good practice in this situation is to run the estimator from different starting points to determine whether multiple fixed points exist. If multiple fixed points are detected, we recommend choosing the one minimizing the sample residual variance, as in \eqref{eq:tilde_beta}. If multiple fixed points are not found across a range of starting values, this increases confidence that the resulting estimate is the global minimizer of the residual variance among fixed points of $\Phi_N$. 
In the simulation study below (Section~\ref{sec:robustness_to_init_sim}), we tested the microPURC estimator extensively with different starting points and found no evidence that convergence to a different fixed point is problematic. Likewise, in our empirical application (Section~\ref{subsec:emp-init-robust}), estimating with different starting values did not reveal any such issues.

We proceed to establish two theorems concerning the convergence and statistical properties of the microPURC estimator. We first show that there is a neighborhood where $\Phi_N$ with probability approaching one admits a unique fixed point, and from which Algorithm~\ref{iterated_algo} is guaranteed to converge to this fixed point. Convergence is quadratic, except for a small error term that goes to zero as the sample size increases.

\begin{theorem}\label{thm:quadratic_convergence}
Suppose Assumption~\ref{as:regularity} holds, then the following holds with probability approaching $1$:
\begin{itemize}
    \item[1.] $\Phi_N$ admits a unique fixed point $\tilde\beta  \in B \subset\mathcal{B}_0$ satisfying $\tilde\beta = \Phi_N(\tilde\beta)$, on some closed ball $B$ centered at $\beta_0$.
\end{itemize}
Additionally, let $\{\tilde\beta_{(t)}\}_{t\ge 0}$ be generated by
Algorithm~\ref{iterated_algo} from a starting value $\tilde\beta_{(0)}$ sufficiently close to $\tilde\beta$.
Then,
\begin{itemize}
    \item[2.] There exists a constant $C>0$ and a random variable $\epsilon_N\ge 0$ such that,
\begin{equation}\label{eq:aq-bound-app}
  \bigl\|\tilde\beta_{(t+1)}-\tilde\beta\bigr\|
  \;\le\;
  C\,\bigl\|\tilde\beta_{(t)}-\tilde\beta\bigr\|^2
  +
  \epsilon_N\,\bigl\|\tilde\beta_{(t)}-\tilde\beta\bigr\|,
\end{equation}
and $\epsilon_N \xrightarrow{p} 0$.
In particular, $\tilde\beta_{(t)}$ converges to $\tilde\beta$ almost quadratically.
\end{itemize}
\end{theorem}

The line search helps improve numerical stability by guiding the iterates towards the fixed point. When the algorithm has reached a point sufficiently close to the fixed point $\tilde{\beta}$, the proof of Theorem~\ref{thm:quadratic_convergence} shows that the non-monotone line search gives step size $\alpha_{(t)} = 1$ for all subsequent iterates, such that the line search ceases to play a role and we retain local quadratic convergence.

The next theorem establishes consistency and asymptotic normality of the estimate $\tilde\beta$. The asymptotic variance has the classical sandwich form $H^{-1}SH^{-1}$, as in standard GMM or quasi-maximum likelihood estimation. Here $H=\Ebb\left[Z^\top \hat P Z\right]$ plays the role of the curvature, or information, matrix familiar from OLS ($X^\top X$) or MLE (the expected Hessian): it captures how strongly the projected link attributes respond to changes in $\beta$. $S$ is the variance of the estimating equation itself, capturing how the deviation of each trip's realized route $y$ from its predicted flow $\hat x$ propagates into noise in the moment condition; the sandwich form then converts this moment-level noise into parameter-level uncertainty. This allows us to compute standard errors and conduct inference on the parameter estimates obtained.

\begin{theorem}
\label{thm:normality}  Given Assumption \ref{as:regularity}, the estimator $\tilde \beta$ is consistent, $\left\| \tilde \beta-\beta_0\right\|\overset{p}%
{\rightarrow} 0$ and $\sqrt{N}$-asymptotically normal $\sqrt{N}(\tilde \beta_N -\beta_0)\overset{D}{\rightarrow} \mathcal{N} (0,H^{-1}S H^{-1})$ where
\begin{eqnarray*}
H=\Ebb\left[ Z^\top \hat P Z\right], \quad S=\Ebb\left[ Z^\top
\hat P \nabla^2 F(x)  (y-x)(y-x)^\top \nabla^2 F(x)\hat P  Z\right].
\end{eqnarray*}
The asymptotic variance of $\tilde \beta$ can be consistently estimated by replacing $H$ and $S$ by their sample counterparts,
\begin{eqnarray}  \label{eq:H_hat and S_hat} \tilde H=\Ebb_N\left[Z^{n\top} \hat P^n Z^n \right], \quad
\tilde S=\Ebb_N\left[ Z^{n\top}\hat P^n \nabla^2 F(\hat x^n) (y^n-\hat x^n)(y^n-\hat x^n)^\top  \nabla^2 F(\hat x^n) \hat P^n Z^n \right],
\end{eqnarray}
where  $\hat x^n=\hat x(Z^n \tilde \beta,b^n)$ and $\hat P^n=\hat P( \hat x^n)$.
\end{theorem}

A consistent estimator of the standard error of $\tilde \beta$ is then
\begin{equation}\label{eq:hat_sigma}
    \tilde \sigma=\sqrt{\mathrm{diag}(\tilde H^{-1}\tilde S \tilde H^{-1})/N}
\end{equation}

\section{Simulation}\label{sec:sim}
This section examines the finite-sample behavior of the microPURC estimator and Algorithm~\ref{iterated_algo} on a benchmark network.

\subsection{Setup}
Throughout the simulation experiments, we use the Eastern Massachusetts network~\citep{bargeraNetwork}, which consists of $|\mathcal{L}|=258$ directed links and 1113 OD pairs. We take as given the corresponding road network and the origin-destination (OD) demand matrix.
We consider a link attribute matrix $Z \in \mathbb{R}^{|\mathcal{L}| \times 3}$ with three link attributes: link length, link travel time, and a synthetic attribute generated as $\exp(\varepsilon_{ij})$ with $\varepsilon_{ij}$ sampled from a normal distribution $\mathcal{N}(0,0.2^2)$.
The true parameter is set to be
\[
  \beta_0 = (0.5, 0.5, 0.5)^\top.
\]
With these parameters, the cost components are of comparable magnitude, and the predicted flow contains several alternative paths for most OD pairs.
To simulate a trip, we first draw an OD pair from the OD matrix with sampling probabilities proportional to the OD demands.
Next, we solve the PURC problem with quadratic perturbation~\eqref{eq:quadF} under $\beta_0$, and then sample a simple path from the resulting optimal link flows. Specifically, starting from the origin, we perform random walks with the following link choice probability
    \begin{equation*}
        p_{ij}^n = \frac{x_{ij}^{n*}}{\sum_{k:(i,k) \in \mathcal{L}} x_{ik}^{n*}},
    \end{equation*}
and terminate at the destination node.
This delivers a set of route choices $\{y_i\}_{i=1}^N$ and associated demand vectors $\{b_i\}_{i=1}^N$, which we use as input to the microPURC estimator. Further details on the simulation setup are provided in  Appendix~\ref{app:implementation_details}. The appendix also reports a sensitivity analysis for the parameters of Algorithm~\ref{iterated_algo}. 
We find that the estimates are stable regardless of the choice of algorithm parameters. The parameters have only limited influence on runtime, so we fix them at their defaults throughout.

\subsection{Finite-sample performance}
We first assess bias, root mean-square error (RMSE), and the quality of the standard errors.
For each sample size
\[
  N \in \{200, 500, 1000, 2000, 5000, 10000, 20000\},
\]
we generate $R=500$ sets of observations $\{y_i\}_{i=1}^N, \{b_i\}_{i=1}^N$, which results in a total of $3,500$ datasets. On each dataset, we estimate $\tilde{\beta}$ from the same initialization
\[
\tilde \beta_{(0)} = (1.0, 1.0, 1.0)^\top.
\]
For parameter $k$ and sample size $N$, denoted as $\tilde \beta^k_{N}$, we compute the Monte Carlo bias and RMSE as follows
\[
  \text{Bias}^k(N) = \mathbb{E}[\tilde\beta_{N}^k - \beta_{0}^k],
  \qquad
  \text{RMSE}^k(N) = \sqrt{\mathbb{E}[(\tilde\beta_{N}^k - \beta_{0}^k)^2]},
\]
and we report their empirical counterparts based on the $R$ replications. The statistics reported in Table~\ref{tab:sim-bias-rmse} and Table~\ref{tab:sim-sd-se} are computed over all replications. In Appendix~\ref{app:sim_converged}, we show that restricting them to the converged runs only leaves the statistics effectively unchanged.
Table~\ref{tab:sim-bias-rmse} presents the estimated bias, RMSE, and $\sqrt{N}$-\text{RMSE} for each coefficient.

\begin{table}[H]
  \centering
  \caption{Bias, RMSE, and $\sqrt{N}$-RMSE by coefficient and sample size.}
  \label{tab:sim-bias-rmse}
  \small
  \setlength\tabcolsep{4pt}
  \begin{tabular}{cccccccccc}
    \toprule
    $N$ & \multicolumn{3}{c}{$\tilde\beta^{\text{length}}$} & \multicolumn{3}{c}{$\tilde\beta^{\text{time}}$} & \multicolumn{3}{c}{$\tilde\beta^{\text{sim}}$} \\
    \cmidrule(lr){2-4} \cmidrule(lr){5-7} \cmidrule(lr){8-10}
    & Bias & RMSE & $\sqrt{N}$-RMSE & Bias & RMSE & $\sqrt{N}$-RMSE & Bias & RMSE & $\sqrt{N}$-RMSE \\
    \midrule
    200 & 0.0229 & 0.1686 & 2.3851 & 0.0170 & 0.1424 & 2.0141 & -0.0029 & 0.4394 & 6.2139 \\
    500 & 0.0115 & 0.0974 & 2.1785 & 0.0050 & 0.0795 & 1.7777 & 0.0027 & 0.2630 & 5.8806 \\
    1,000 & 0.0093 & 0.0654 & 2.0683 & -0.0021 & 0.0552 & 1.7458 & 0.0061 & 0.1991 & 6.2955 \\
    2,000 & 0.0003 & 0.0441 & 1.9741 & 0.0021 & 0.0384 & 1.7187 & 0.0056 & 0.1302 & 5.8238 \\
    5,000 & 0.0007 & 0.0277 & 1.9622 & 0.0012 & 0.0230 & 1.6239 & -0.0007 & 0.0801 & 5.6605 \\
    10,000 & 0.0008 & 0.0177 & 1.7725 & 0.0001 & 0.0168 & 1.6790 & 0.0014 & 0.0571 & 5.7122 \\
    20,000 & -0.0004 & 0.0137 & 1.9314 & 0.0003 & 0.0115 & 1.6247 & 0.0012 & 0.0404 & 5.7162 \\
    \bottomrule
  \end{tabular}
\end{table}

The estimator is effectively unbiased at all sample sizes: the absolute bias is an order of magnitude smaller than the corresponding RMSE for every coefficient, implying the bias is negligible compared to random error.
The $\sqrt{N}\,\text{RMSE}^k(N)$ is approximately constant across values of $N$, which is consistent with the $\sqrt{N}$-rate predicted by the asymptotic theory.

To assess the approximation in the standard deviation formula, we compare the empirical standard deviation of $\tilde \beta^k_{N}$ across replications to the average estimated standard error using Eq.~\eqref{eq:hat_sigma}, and we examine the coverage of the Wald intervals.
Let
\[
  \text{SD}^k(N) = \sqrt{\frac{1}{R-1}\sum_{r=1}^{R}\left(\tilde\beta_{N,(r)}^{k}-\mathbb{E}[\tilde{\beta}^k_{N}]\right)^2}, \qquad
  \overline{\text{SE}}^k(N) = \mathbb{E}[\tilde{\sigma}^k_{N}],
\]
and consider the nominal $95\%$ interval
\[
  \tilde\beta^k_{N} \pm 1.96\,\tilde{\sigma}^k_{N}.
\]
Table~\ref{tab:sim-sd-se} reports $\text{SD}^k(N)$, $\overline{\text{SE}}^k(N)$, and the empirical coverage, that is, the percentage of replications in which the nominal $95\%$ interval contains $\beta_0$.

\begin{table}[H]
  \centering
  \caption{Empirical SD, $\overline{\text{SE}}$, and 95\% coverage by coefficient and sample size.}
  \label{tab:sim-sd-se}
  \small
  \begin{tabular}{cccccccccc}
    \toprule
    $N$ & \multicolumn{3}{c}{$\tilde\beta^{\text{length}}$} & \multicolumn{3}{c}{$\tilde\beta^{\text{time}}$} & \multicolumn{3}{c}{$\tilde\beta^{\text{sim}}$} \\
    \cmidrule(lr){2-4} \cmidrule(lr){5-7} \cmidrule(lr){8-10}
    & SD & $\overline{\text{SE}}$ & Coverage & SD & $\overline{\text{SE}}$ & Coverage & SD & $\overline{\text{SE}}$ & Coverage \\
    \midrule
    200 & 0.1673 & 0.1333 & $92.18\%$ & 0.1415 & 0.1180 & $90.58\%$ & 0.4398 & 0.4063 & $94.19\%$ \\
    500 & 0.0968 & 0.0844 & $90.20\%$ & 0.0794 & 0.0744 & $93.60\%$ & 0.2632 & 0.2550 & $93.40\%$ \\
    1,000 & 0.0648 & 0.0599 & $93.40\%$ & 0.0552 & 0.0524 & $94.20\%$ & 0.1992 & 0.1800 & $91.80\%$ \\
    2,000 & 0.0442 & 0.0418 & $94.40\%$ & 0.0384 & 0.0369 & $94.20\%$ & 0.1302 & 0.1268 & $94.40\%$ \\
    5,000 & 0.0278 & 0.0268 & $93.40\%$ & 0.0230 & 0.0233 & $96.20\%$ & 0.0801 & 0.0803 & $94.40\%$ \\
    10,000 & 0.0177 & 0.0189 & $97.00\%$ & 0.0168 & 0.0164 & $94.00\%$ & 0.0572 & 0.0568 & $94.40\%$ \\
    20,000 & 0.0137 & 0.0134 & $93.80\%$ & 0.0115 & 0.0116 & $94.60\%$ & 0.0404 & 0.0402 & $95.40\%$ \\
    \bottomrule
  \end{tabular}
\end{table}

Across all coefficients and sample sizes in Table~\ref{tab:sim-sd-se}, the empirical standard deviations and the average estimated standard errors are very close.
The resulting 95\% Wald intervals achieve coverage close to the nominal level: for all three coefficients and all $N$, the empirical coverage lies in the 90–97\% range, and there is no systematic drift as $N$ increases.
These patterns are in line with the asymptotic variance expression in Eq.~\eqref{eq:hat_sigma} and with the asymptotic normality predicted by Theorem~\ref{thm:normality}.

Algorithmic diagnostics for the fixed-point iteration are reported in Table~\ref{tab:sim-runtime}.
\begin{table}[H]
  \centering
\begin{threeparttable}
  \caption{Convergence and runtime of Algorithm~\ref{iterated_algo} by sample size.}
  \label{tab:sim-runtime}
  \begin{tabular}{cccccccc}
    \toprule
    & \multicolumn{7}{c}{$N$} \\
    \cmidrule(lr){2-8}
    & 200 & 500 & 1,000 & 2,000 & 5,000 & 10,000 & 20,000 \\
    \midrule
    Convergence rate$^*$ & $75.00\%$ & $80.20\%$ & $88.40\%$ & $92.40\%$ & $94.40\%$ & $97.80\%$ & $99.20\%$ \\
    Avg. iterations & 6.96 & 6.30 & 6.40 & 6.29 & 5.64 & 4.86 & 4.26 \\
    Avg. runtime[s] & 1.87 & 1.84 & 1.29 & 1.23 & 1.40 & 0.91 & 0.73 \\
    \bottomrule
  \end{tabular}
  \begin{tablenotes}
      \footnotesize
      \item $*$: Convergence rate is defined as the percentage of replication runs that reach fixed-point residual tolerance within 200 iterations.
    \end{tablenotes}
  \end{threeparttable}
\end{table}

Convergence rates are high and increase with $N$, reaching about 99\% for large samples.
On the other hand, convergence rates are lower for small $N$.
This is because, for small $N$, the empirical map $\Phi_N$ can be noisy and may not behave like the population map $\Phi$ outside a neighborhood of $\beta_0$; the algorithm hits the iteration cap (set as 200 here) and is declared non-convergent under the strict tolerance $\tau = 10^{-4}\sqrt{K}$.
As $N$ grows, $\Phi_N$ stabilizes and the observed convergence rate increases.
We formalize these properties of the mappings $\Phi$ and $\Phi_N$ in Lemma~\ref{lem:Phi_properties} in the Appendix.

The average number of iterations reduces from roughly 7 to about 4 as $N$ increases. In this example, the runtime does not grow with $N$: although the cost of solving the PURC problems (Step 1) grows linearly in $N$, the runtime is largely compensated by the decreases in the number of iterations for updating $\tilde \beta$. We exploit in our implementation the fact that the microPURC estimator is embarrassingly parallelizable.
We show the scalability of the microPURC estimator in Section~\ref{sec:scalability} and demonstrate its computational efficiency in the large-scale empirical application in Section~\ref{sec:empirical_application}.

\subsection{Robustness to initialization}\label{sec:robustness_to_init_sim}
We next examine how sensitive the proposed Algorithm~\ref{iterated_algo} is to the choice of starting values.
Our Theorem~\ref{thm:quadratic_convergence} and Theorem~\ref{thm:normality} characterize the (local) behavior of the fixed-point map in a neighborhood of $\beta_0$, but in practice one would like the algorithm to converge reliably from a wide range of initializations.
To explore the size of this neighborhood in practice, we vary the starting point $\tilde \beta_{(0)}$ randomly and track the convergence of the estimator and the final estimate.

We fix $N = 10,000$ and generate $10$ independent datasets in the same way as above.
Working in log-parameter space $\theta = \log \beta$, we consider balls of radii
\[
  r \in \{0.5, 1.0, 2.0, 3.0, 5.0, 8.0\}
\]
around $\theta_0 = \log \beta_0$.
For each dataset $d$ and each radius $r$, we draw $M = 1000$ initial values $\theta$ uniformly from the Euclidean ball of radius $r$ around $\theta_0$, exponentiate to obtain positive starting values $\tilde \beta_{(0)} = \exp(\theta)$, and run Algorithm~\ref{iterated_algo} on dataset $d$ with initial parameter $\tilde \beta_{(0)}$. This results in a total of $60,000$ estimation runs.

For each run, we record whether the algorithm converges (reaching $\tau=10^{-4}\sqrt{K}$ within 200 iterations), the number of iterations, and the Euclidean distance between the final estimate and the true parameter, $\norm{\tilde\beta-\beta_0}$.
Table~\ref{tab:sim-cball} summarizes these diagnostics averaged over datasets and initializations.

\begin{table}[H]
  \centering
  \begin{threeparttable}
  \caption{Convergence and accuracy by radius of the initialization ball.}
  \label{tab:sim-cball}
  \small
  \begin{tabular}{ccccc}
    \toprule
    Radius $r$
      & Starting distance
      & Accuracy
      & Convergence rate$^*$
      & Avg. iterations \\
      & $\mathbb{E}\big[\|\tilde\beta_{(0)}-\beta_0\|\big]$
      & $\mathbb{E}\big[\|\tilde\beta-\beta_0\|\big]$
      &
      &  \\
    \midrule
     0.5 & 0.1267 & 0.0557 & 100.0\% &  3.73 \\
     1.0 & 0.2632 & 0.0557 & 100.0\% &  4.01 \\
     2.0 & 0.6246 & 0.0557 & 100.0\% &  4.36 \\
     3.0 & 1.2200 & 0.0557 & 100.0\% &  4.68 \\
     5.0 & 4.7606 & 0.0557 & 100.0\% & 5.53 \\
     8.0 &38.4247 & 0.0556 & 99.81\% &  6.88\\
    \bottomrule
  \end{tabular}
  \begin{tablenotes}
      \footnotesize
      \item $*$: Convergence rate is defined as the percentage of replication runs that reach fixed-point residual tolerance $\tau = 10^{-4}\sqrt{K}$ within 200 iterations.
    \end{tablenotes}
  \end{threeparttable}
\end{table}

The starting values span a wide range: at $r = 0.5$ a typical initialization is very close to $\beta_0$, with $\mathbb{E}\|\tilde\beta_{(0)}-\beta_0\|\approx 0.13$, while at $r = 3.0$ it is already about $1.2$ in Euclidean norm, and at $r = 8.0$ the average distance is roughly $38$, many times the length of the true parameter vector ($\norm{\beta_0} \approx 0.87$). 
Despite this, the algorithm's behavior remains stable. 
The algorithm converges from every start at $r \leq 5.0$, and from 99.8\% of starts
even at $r = 8$.
Among converged runs, the average distance to the truth, $\mathbb{E}\|\tilde\beta - \beta_0\|$, is essentially constant across radii at about $0.056$, and the average number of iterations increases only moderately with $r$, remaining below about 7 even at the largest radii. These patterns indicate that, in this example, over a wide range of initial values in log-parameter space, including starting points that are very far from $\beta_0$, the iterations generated by Algorithm~\ref{iterated_algo} converge with high probability to essentially the true parameter $\beta_0$.

\subsection{Active set and convergence analysis}\label{sec:active_set}
This subsection demonstrates the local stability of the trip-level active sets as required by the asymptotic theory (Assumption~\ref{as:regularity}(e)) and connects it to the active-set changes observed during estimation.

A trip's active set changes only when one of its links switches from active to inactive. Under the dual slack classification, this switch occurs when the link slack
\[
    s_{ij}=c_{ij}+(A^\top\hat\eta)_{ij}
\]
crosses zero. 
Since the slack is measured in cost units, dividing by $\|Z_{ij}\|$ converts it to the units of $\beta$. 
This allows us to define the margin of trip $n$ at the estimate as
\begin{equation}\label{eq:margin_def}
  m_n \coloneqq \min_{ij}\ \frac{|s_{ij}|}{\|Z_{ij}\|},
\end{equation}
where the minimum is taken over all links.
Thus $m_n$ can be interpreted as the distance, in $\beta$-units, from $\tilde\beta$ to the nearest active-set boundary for trip $n$. 
A small margin means that the trip is close to a switching boundary, while a large margin means that the trip lies well inside a region where its active set is fixed.

\begin{figure}[H]
  \centering
  \includegraphics[width=0.72\textwidth]{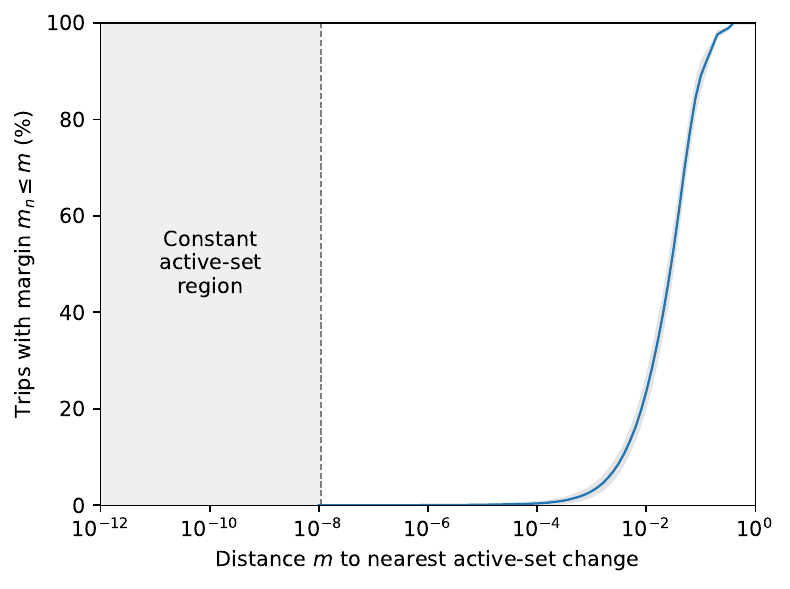}
  \caption{Cumulative distribution of the per-trip active-set margin at the estimate ($N=500$ with shaded band representing $R=500$ replications).}
  \label{fig:margin-cdf}
\end{figure}

Figure~\ref{fig:margin-cdf} shows the distribution of $m_n$. First, we see that the margin is bounded away from zero: the minimum over the $R=500$ replications ($250{,}000$ trips) is about $10^{-8}$, so, even at this scale, no trip lies on an exact boundary, and every trip has a genuine neighborhood of $\tilde\beta$ on which its active set is fixed. Hence, it provides a direct numerical validation of Assumption~\ref{as:regularity}(e). 
The shaded region sharpens this from a per-trip to a sample-wide observation: below the smallest margin ($1\times10^{-8}$ in this example\footnote{We use machine precision at $10^{-16}.$}) no trip changes its active links at all, so every $\hat B^n$ and projection $\hat P^n$ is constant and the sample map $\Phi_N$ is smooth in this region. Second, on the other hand, the neighborhoods vary across trips: the median trip lies about $0.03$ $\beta$-units from its nearest boundary, but roughly a quarter lie within $0.01$ and a thin tail runs down to the $10^{-8}$ floor. This means local constancy holds in general as the theory predicts, but its radius can be trip-specific.

\begin{figure}[H]
  \centering
  \includegraphics[width=0.72\textwidth]{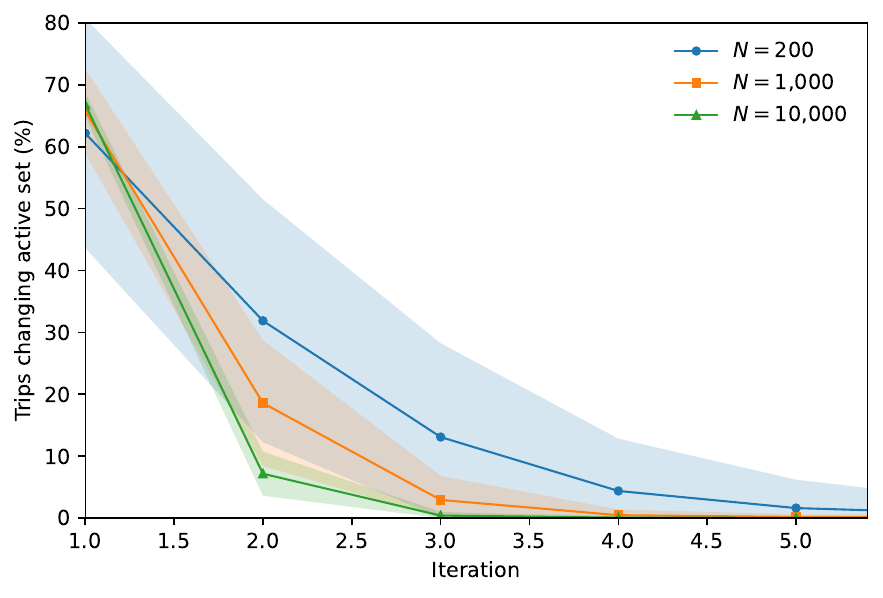}
    \caption{Active-set change across the iterations of Algorithm~\ref{iterated_algo}, by sample size (with shaded band representing $R=500$ replications).}
  \label{fig:iter_active_set}
\end{figure}

This margin distribution explains the pattern of active-set changes during estimation. 
Figure~\ref{fig:iter_active_set} follows the active sets across the outer iterations of Algorithm~\ref{iterated_algo}, plotting at each iteration the share of trips whose active set differs from the previous iterate. 
These changes are substantial in the first iterations because the early parameter updates move far enough to cross many active-set boundaries. They then quickly disappear in larger samples: after a few iterations, the iterates enter a region where all trip-level active sets are effectively stable. At $N=200$, however, a minority of replications retain a small persistent tail. This tail is carried by a small number of trips (less than $0.5\%$) with small margins, whose active sets switch back and forth as the iterate crosses their nearby boundaries.

\begin{figure}[H]
  \centering
  \includegraphics[width=1.0\textwidth]{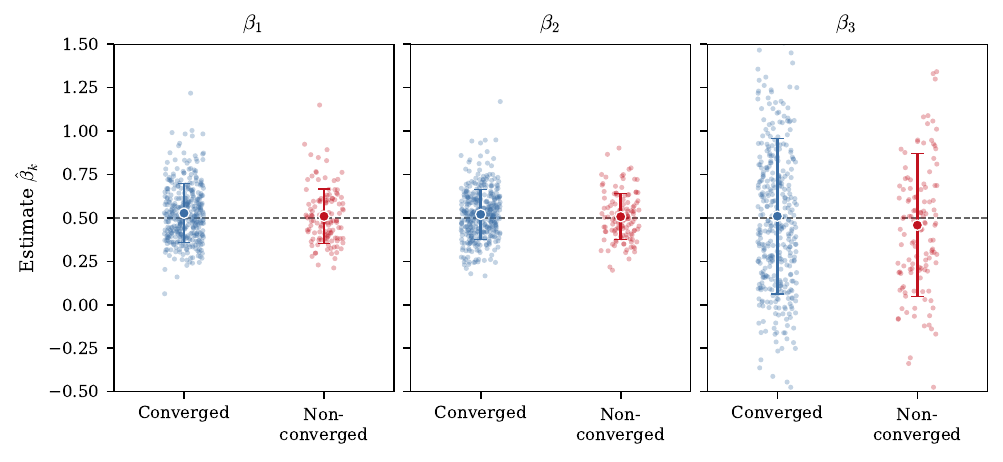}
    \caption{Estimates at $N=200$ by convergence status with $R=500$ replications (dashed line represents true value).}
  \label{fig:nonconv-estimates}
\end{figure}

Crucially, failing to converge leaves the estimates unbiased. Figure~\ref{fig:nonconv-estimates} plots the estimates at $N=200$, the worst case in our examples, separately for the converged and non-converged replications. 
For every coefficient, the two groups share the same center: the mean of the non-converged estimates is statistically indistinguishable from both the converged mean and the truth, so non-convergence introduces no systematic error. 
In addition, the standard deviation of the non-converged estimates is about $8\%$ smaller than that of the converged ones for each coefficient. 
Read together with Figure~\ref{fig:iter_active_set}, this indicates that non-convergence is driven by a small number of trips lying close to an activation boundary, which has little impact on statistical inference in practice.

In sum, the local constant active-set condition is empirically satisfied, and once the iterate enters a common active-set region, the estimator is locally well behaved. Non-convergence arises only when the finite-sample iterate stays close enough to a boundary that a few small-margin trips cannot settle, while such a finite-sample boundary effect shrinks with $N$.
 
\subsection{Robustness to repeated trips}\label{sec:repeated_trips}

The asymptotic theory of Section~\ref{sec:estimation} requires trips are i.i.d.~(Assumption~\ref{as:regularity}(a) as in classic route choice literature. 
In practice, however, route choice datasets often contain repeated trips by the same traveler in violation of the i.i.d. assumption. 
Consequently, a sample of $N$ trips carries the information of fewer than $N$ independent observations. 
We now illustrate how this affects statistical inference. 

We simulate repeated trips as follows. Each of $J$ travelers is assigned one origin-destination pair, drawn as in Section~\ref{sec:sim}, and one habitual route through it, drawn once from the route distribution induced by the true parameter $\beta_0$. The traveler then repeats that route on all $\mathcal{T}$ of their trips, so all within-traveler variation is eliminated. 
We fix $N = J \times \mathcal{T} = 10{,}000$, so that the dataset satisfies the i.i.d. assumption when $\mathcal{T}=1$.

For each $\mathcal  T$, we run $R=500$ replications and, at each estimate, compare two standard errors. 
The naive standard error is  Eq.~\eqref{eq:hat_sigma}, whose meat $\tilde S=\Ebb_N[\hat g^n\hat g^{n\top}]$ treats the per-trip scores $\hat g^n=Z^{n\top}\hat P^n\nabla^2 F(\hat x^n)(y^n-\hat x^n)$ as independent. The traveler-clustered standard error sums these scores within each traveler before forming the meat,
\begin{equation}\label{eq:S_cluster}
  \tilde S_{\mathrm{cl}}=\frac{J}{J-1}\,\frac{1}{N}\sum_{j=1}^{J}\Big(\sum_{n\in j}\hat g^n\Big)\Big(\sum_{n\in j}\hat g^n\Big)^{\!\top},
  \qquad
  \tilde\sigma_{\mathrm{cl}}=\sqrt{\mathrm{diag}\!\big(\tilde H^{-1}\tilde S_{\mathrm{cl}}\tilde H^{-1}\big)/N},
\end{equation}
where $n$ ranges over the trips of traveler $j$, the factor $J/(J-1)$ is the finite-sample correction, $\tilde H$ is the bread defined in Eq.~\eqref{eq:H_hat and S_hat}, and the Wald interval uses a $t_{J-1}$ critical value. Results are reported in Table~\ref{tab:sim-repeated}.

\begin{table}[H]
  \centering
  \begin{threeparttable}
  \caption{Robustness to repeated trips ($N=10{,}000$, $R=500$ replications).}
  \label{tab:sim-repeated}
  \setlength\tabcolsep{12pt}
  \begin{tabular}{lrrrrr}
    \toprule
    & \multicolumn{5}{c}{Trips per traveler $\mathcal T$} \\
    \cmidrule(lr){2-6}
    & $1$ & $2$ & $5$ & $10$ & $20$ \\
    Number of travelers $J$ & ($10{,}000$) & ($5{,}000$) & ($2{,}000$) & ($1{,}000$) & ($500$) \\
    \midrule
    \multicolumn{6}{l}{\textit{Bias}} \\
    \midrule
    $\tilde\beta^{\text{length}}$ & $-0.0005$ & $0.0014$ & $0.0008$ & $0.0018$ & $0.0075$ \\
    $\tilde\beta^{\text{time}}$   & $0.0001$ & $-0.0002$ & $0.0027$ & $0.0044$ & $0.0031$ \\
    $\tilde\beta^{\text{sim}}$    & $0.0031$ & $-0.0031$ & $0.0001$ & $0.0070$ & $0.0173$ \\
    \addlinespace
    \midrule
    \multicolumn{6}{l}{\textit{Monte-Carlo SD}} \\
    \midrule
    $\tilde\beta^{\text{length}}$ & $0.0185$ & $0.0276$ & $0.0402$ & $0.0623$ & $0.0935$ \\
    $\tilde\beta^{\text{time}}$   & $0.0164$ & $0.0242$ & $0.0378$ & $0.0550$ & $0.0784$ \\
    $\tilde\beta^{\text{sim}}$    & $0.0546$ & $0.0859$ & $0.1294$ & $0.1912$ & $0.2775$ \\
    \addlinespace
    \midrule
    \multicolumn{6}{l}{\textit{Naive $95\%$ coverage}} \\
    \midrule
    $\tilde\beta^{\text{length}}$ & $95.2\%$ & $80.2\%$ & $64.4\%$ & $49.4\%$ & $34.3\%$ \\
    $\tilde\beta^{\text{time}}$   & $95.6\%$ & $82.8\%$ & $61.6\%$ & $44.8\%$ & $32.9\%$ \\
    $\tilde\beta^{\text{sim}}$    & $95.4\%$ & $80.4\%$ & $61.8\%$ & $44.8\%$ & $30.9\%$ \\
    \addlinespace
    \midrule
    \multicolumn{6}{l}{\textit{Traveler-clustered $95\%$ coverage}} \\
    \midrule
    $\tilde\beta^{\text{length}}$ & $95.2\%$ & $94.2\%$ & $95.4\%$ & $92.6\%$ & $94.0\%$ \\
    $\tilde\beta^{\text{time}}$   & $95.6\%$ & $94.8\%$ & $94.8\%$ & $93.2\%$ & $93.6\%$ \\
    $\tilde\beta^{\text{sim}}$    & $95.4\%$ & $93.2\%$ & $94.4\%$ & $93.0\%$ & $92.8\%$ \\
    \bottomrule
  \end{tabular}
  \end{threeparttable}
\end{table}

Table~\ref{tab:sim-repeated} confirms two points. Repeated trips do not bias the estimator: the bias stays statistically indistinguishable from zero at every $\mathcal T$, so even when the $10{,}000$ trips come from as few as $J=500$ travelers, microPURC recovers the true $\beta_0$. 
On the other hand, the coverage of the naive standard error falls from $95\%$ at $\mathcal  T=1$ to about a third ($31$--$34\%$) at $\mathcal  T=20$, while the traveler-clustered standard error tracks the true dispersion and holds $93$--$95\%$ throughout. 

In sum, repeated trips do not bias microPURC; they inflate the true sampling variance, and the traveler-clustered standard error tracks that inflation. The two standard errors play asymmetric roles. The clustered estimator remains consistent under arbitrary within-traveler dependence, including none: when trips are independent, the within-traveler cross-products $\hat g^n\hat g^{m\top}$, $n\neq m$, vanish in expectation and $\tilde S_{\mathrm{cl}}$ estimates the same limit as $\tilde S$, as the $\mathcal T=1$ column of Table~\ref{tab:sim-repeated} shows. Its cost is precision rather than level, since it is effectively an average over $J$ rather than $N$ terms, and it is correspondingly noisy when the number of travelers is small. The i.i.d.\ formula, by contrast, is consistent only when the within-traveler scores are uncorrelated; under the positive within-traveler dependence one expects when travelers repeat habitual routes, it understates the sampling uncertainty. We therefore recommend clustering by traveler whenever identifiers are available. When identifiers are unavailable, as is common in route choice datasets, standard errors computed under independence should be read as a lower bound on the true sampling uncertainty.

\subsection{Comparison with alternative estimators}\label{sec:estimator-comparison}

We now compare the microPURC estimator with three alternative estimators on the same simulated data. 
\begin{itemize}
  \item \textbf{PURC-OLS} \citep{fosgerau_perturbed_2022}, the aggregate estimator of the  PURC model. It averages discrete trips to the level of origin-destination flows, determines each active set from the observed OD flow $\bar y$, and solves the projected first-order condition~\eqref{eq:projectedFOC}, obtaining parameter estimates with ordinary least squares (OLS). Unlike microPURC, it does not need to solve any flow-prediction PURC problem, but cannot work with individual-level trip data.
  \item \textbf{MLE}, which computes the (Markov) path likelihood from joint link choice probabilities Eq.~\eqref{eq:markov_sampling} for the observed routes. 
  We then maximize the (non-smooth) likelihood with Nelder--Mead search~\citep[][scipy implementation]{gao2012implementing}.
  \item \textbf{Direct-RV}, which minimizes the residual variance $RV_N(\beta)$~\eqref{eq:RV} directly with the same Nelder--Mead search method.
\end{itemize}
As in microPURC, we use the efficient quadratic solver~\citep{schwan2023piqp} for the PURC-flow prediction problem for both the MLE and direct-RV estimators.
Table~\ref{tab:sim-estimators} reports, for each estimator and sample size, the RMSE of $\tilde\beta$, the median runtime per fit, and, for the MLE, the share of replications that fail due to zero probability. We initialize all estimators with $\tilde \beta_{(0)} = \mathbf{1}$.

\begin{table}[H]
  \centering
  \begin{threeparttable}
  \caption{Comparison with alternative estimators on the EMA network ($R=100$ replications per sample size, $\beta_0=(0.5,0.5,0.5)$).}
  \label{tab:sim-estimators}
  \small
  \setlength\tabcolsep{10pt}
  \begin{tabular}{lrrrrrrr}
    \toprule
    Estimator & \multicolumn{7}{c}{Sample size $N$} \\
    \cmidrule(lr){2-8}
     & 200 & 500 & 1{,}000 & 2{,}000 & 5{,}000 & 10{,}000 & 20{,}000 \\
    \midrule
    \multicolumn{8}{l}{\textit{RMSE$^*$}} \\
    \midrule
    microPURC & 0.542 & 0.283 & 0.243 & 0.124 & 0.090 & 0.062 & 0.048 \\
    PURC-OLS & 1.165 & 0.733 & 0.464 & 0.291 & 0.181 & 0.116 & 0.078 \\
    MLE & 0.773 & 0.689 & 0.685 & 0.726 & 0.888 & --- & --- \\
    direct-RV & 0.371 & 0.219 & 0.152 & 0.103 & 0.066 & 0.044 & 0.032 \\
    \addlinespace
    \midrule
    \multicolumn{8}{l}{\textit{Median runtime [s]}} \\
    \midrule
    microPURC & 0.093 & 0.103 & 0.111 & 0.141 & 0.194 & 0.238 & 0.350 \\
    PURC-OLS & 0.029 & 0.048 & 0.068 & 0.096 & 0.145 & 0.194 & 0.265 \\
    MLE & 2.294 & 2.006 & 2.209 & 4.128 & 15.863 & --- & --- \\
    direct-RV & 1.497 & 1.707 & 2.023 & 2.645 & 5.386 & 6.374 & 19.807 \\
    \addlinespace
    \midrule
    \multicolumn{8}{l}{\textit{MLE replications degenerating [\%]}} \\
    \midrule
    MLE & 25 & 42 & 70 & 90 & 98 & 100 & 100 \\
    \bottomrule
  \end{tabular}
  \begin{tablenotes}
      \footnotesize
      \item $^*$: RMSE is computed over replications that returned an estimate, and runtimes are medians over the same set. 
    \end{tablenotes}
  \end{threeparttable}
\end{table}

PURC-OLS is the fastest of the four, since it needs no PURC solve, but its RMSE is
roughly twice microPURC's at every sample size ($0.078$ against $0.048$ at
$N=20{,}000$).
The MLE is least reliable for PURC, and is degenerate as anticipated in Section~\ref{sec:litreview}. Whenever an observed route uses a link carrying no predicted flow (at any iteration), that route has probability zero and the log-likelihood is undefined, so the estimator fails to return an estimate. Such failures worsen as sample size increases: the share of replications that degenerate rises from $25\%$ at $N=200$ to $100\%$ at $N\ge 10{,}000$. 
On the other hand, the residual variance criterion $RV_N$ avoids this failure mode because the RV stays finite on zero-flow links.

However, minimizing $RV_N$ directly is computationally expensive. Although its RMSE is slightly better than microPURC's ($0.032$ versus $0.048$ at $N=20{,}000$), the direct-RV estimator takes more than fifty times as long to run. At $N=20{,}000$, the direct-RV estimator takes about $19.8$ seconds, while it takes only $0.35$ seconds for microPURC. This saving occurs because the microPURC map $\Phi_N$ replaces the line search over $RV_N$ with a single projected regression per iteration, thereby converging in a few iterations.

Figure~\ref{fig:sim-recovery-cdf} plots the share of estimates within a distance from the true parameter $\beta_0$, under various initialization distances as in Section~\ref{sec:robustness_to_init_sim}. For microPURC, the four radius curves lie exactly on top of one another and reach $1$ at a small distance from the truth. 
In comparison, direct-RV rises more steeply, reflecting its slightly higher efficiency, while it incurs a slight drop in performance at radius $r=8$: about $6\%$ of runs settle beyond $0.25$, though all of them still reach the truth eventually. 
The MLE behaves differently: as the radius grows, its curves plateau far below $1$, where, with $r=8$, the curve flattens near $0.4$. The remaining mass corresponds to runs that diverge. Such behavior is consistent with the local convergence property of MLE estimators.

\begin{figure}[H]
  \centering
  \includegraphics[width=0.35\textwidth]{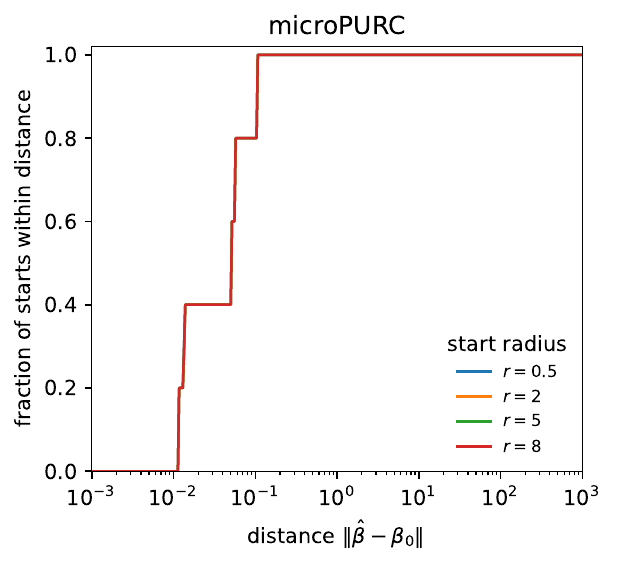}\hfill
  \includegraphics[width=0.315\textwidth]{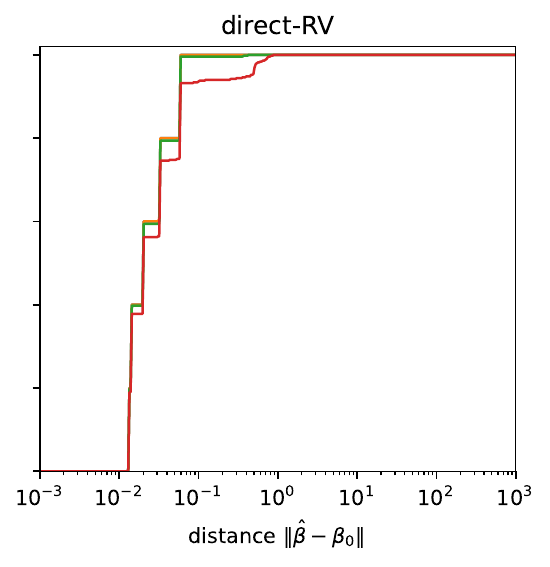}\hfill
  \includegraphics[width=0.315\textwidth]{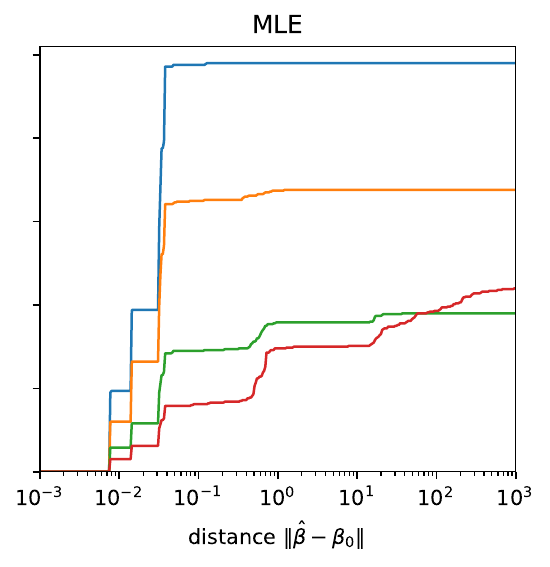}
  \caption{Estimator performance to initialization.}
  \label{fig:sim-recovery-cdf}
\end{figure}

Taken together, microPURC provides better statistical efficiency than the aggregate estimator, avoids the degeneracy that makes the MLE unstable, and attains the robustness of the direct search at lower computational costs.

\section{Computational scalability}\label{sec:scalability}

We now illustrate the computational scalability of the proposed microPURC estimator. 
We first demonstrate it on standard benchmark road networks, then use controlled grid networks to isolate the computational costs along each of the three dimensions: the number of trips $N$, the network size $|\mathcal{L}|$, and the number of parameters $K$. 

Table~\ref{tab:scaling-real} reports estimation on six benchmark networks, each with $N=5{,}000$ simulated trips. 
Estimation takes between $0.09$ seconds on SiouxFalls and $3.4$ seconds on ChicagoSketch, the largest at $2{,}950$ links. 
All estimations complete within four to five outer iterations, while the memory requirement is negligible. 
The table also reports the mean number of active links per trip, which stays small relative to the network: it ranges from $7.9$ on SiouxFalls to $75.6$ on Winnipeg, so that a single trip activates about $10\%$ of the smallest network and under $1\%$ of the largest.
Because each trip's projection is formed on its active subnetwork, this sparsity bounds the dimension of the per-trip linear algebra and hence the cost of an outer iteration.
Together with the empirical application of Section~\ref{sec:empirical_application}, which runs on a network of $98{,}503$ links, this indicates that the estimator is practical on real networks at the scales encountered in applications.

\begin{table}[H]
  \centering
  \begin{threeparttable}
  \caption{Scalability on benchmark road networks ( $N=5{,}000$ trips, over 5 replications).}
  \label{tab:scaling-real}
  \small
  \setlength\tabcolsep{5.5pt}
  \begin{tabular}{lrrrrrrr}
    \toprule
    Network & $|\mathcal{N}|$ & $|\mathcal{L}|$ & $N$ & Mean active links & Runtime [s] & Memory [MB] & Iters \\
    \midrule
    SiouxFalls & 24 & 76 & 5,000 &  7.9 & 0.09 & 0.00  & 4.0 \\
    EMA & 74 & 258 & 5,000 & 12.0 & 0.16 & 0.00  & 4.0 \\
    Anaheim & 416 & 914 & 5,000 & 45.2 &0.54 & 0.38  & 5.0 \\
    Barcelona & 930 & 2,522 & 5,000 & 51.2 &2.78 & 1.88  & 5.0 \\
    Winnipeg & 1,040 & 2,836 & 5,000 & 75.6 &2.94 & 1.83 &  5.0 \\
    ChicagoSketch & 933 & 2,950 & 5,000 & 22.7 &3.41 & 30.34  & 5.0 \\
    \bottomrule
  \end{tabular}
  \end{threeparttable}
\end{table}

To test scalability on each problem dimension, we turn to controlled grid networks (Figure~\ref{fig:grid}). Each link in the network has $K$ synthetic attributes; trips are generated by drawing origin-destination pairs uniformly over the demand nodes, simulating flows at a fixed $\beta_0$, and sampling routes as in Section~\ref{sec:sim}. 

\begin{figure}[H]
  \centering
  \includegraphics[width=0.5\textwidth]{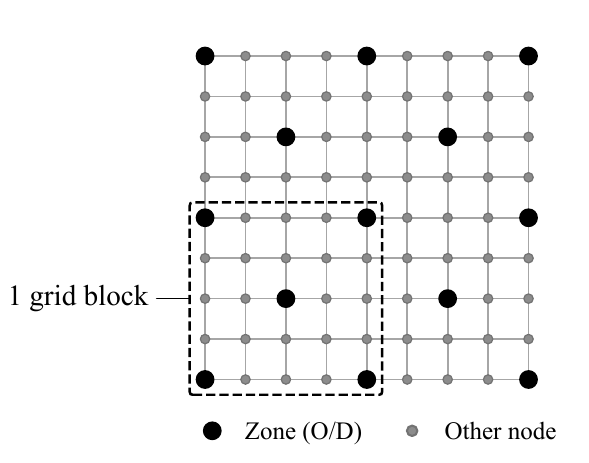}
  \caption{Grid networks.}
  \label{fig:grid}
\end{figure}

Figure~\ref{fig:scaling} reports the scalability with respect to the number of trips and network sizes. 
At a grid network with $|\mathcal{L}| = 1,520$ links, runtime grows sublinearly in the number of trips (left): from $0.34$ seconds at $N=500$ to $10.9$ seconds at $N=50{,}000$ with a log-log slope of $0.79$ below the linear reference. 
Each outer iteration costs $O(N)$, since the $N$ trip-level problems are solved independently, and the number of outer iterations falls from six to four, as the sample grows and $\Phi_N$ stabilizes around $\beta_0$, so the per-trip cost declines. 
As shown in the right panel, with 5000 samples, computation time grows about linearly in network size: from $0.08$ seconds at $80$ links to $21.1$ seconds at $9{,}800$ links. 

\begin{figure}[H]
  \centering
  \begin{subfigure}{0.49\textwidth}
    \centering
    \includegraphics[width=\linewidth]{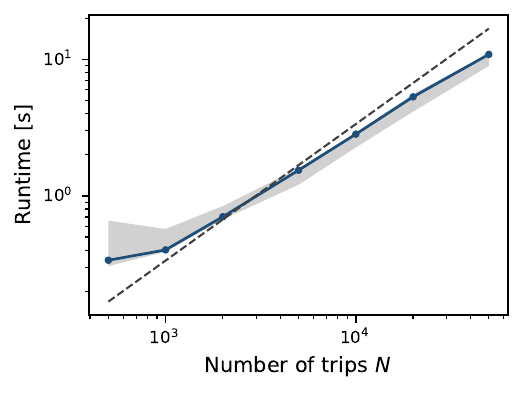}
  \end{subfigure}
  \hfill
  \begin{subfigure}{0.49\textwidth}
    \centering
    \includegraphics[width=\linewidth]{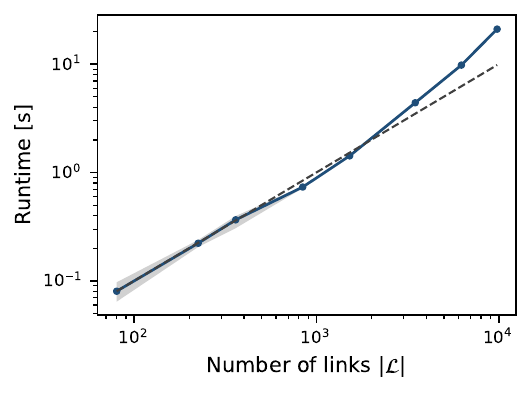}
  \end{subfigure}
  \caption{Runtime on grid networks. (dashed line is a linear reference; the shaded band represents variations over $R=5$ repetitions).}
  \label{fig:scaling}
\end{figure}

Finally, runtime is constant in the number of parameters. On a network with $|\mathcal{L}| = 1,520$ links and $N = 5,000$ samples, the average runtime is $1.47$ seconds for $K$ between 2 and 8, with a standard deviation of 0.12 seconds. See details in Table~\ref{tab:scalability_K} in the Appendix.

\section{Empirical application}\label{sec:empirical_application}

This section applies the microPURC estimator to trip-level route choices observed on the Danish national road network. We find that the estimator is computationally feasible at scale and yields economically meaningful parameter estimates.
We report parameter estimates and standard errors for two specifications and interpret magnitudes using the implied value of travel time, which we benchmark against census wage statistics.

\subsection{Data}
The data are drawn from a Danish road pricing experiment conducted in 2024--2025.
We use observations recorded during the experiment's control periods (the status quo).
Trips are recorded automatically by a dedicated smartphone app and map-matched to the Danish national road network.
Each trip is coded as a link-incidence vector $y\in\{0,1\}^{\mathcal{L}}$.
For each trip, we observe the corresponding origin--destination vector $b$ and the matrix of link attributes $Z$ used to construct link costs $c=Z\beta$.
We retain regular weekday trips with a distance of at least 2 km and a duration of at least 5 minutes, and we winsorize trip distance and duration at the 1\% and 99\% quantiles. This results in a dataset of $N=18{,}019$ trips from 574 unique users, with mean trip distance of 14.72 km and mean trip duration of 18.87 minutes.

The road network contains $98{,}503$ directed links.
Link attributes vary across 10 time intervals per day.
The main attributes are link length (km), GPS-based link travel time (min) when available ($\mathrm{time}_{\mathrm{GPS}}$), free-flow link travel time (min) otherwise  ($\mathrm{time}_{\mathrm{ff}}$), and a junction indicator. We associate observations with time-dependent travel times according to their departure time periods.

Additional descriptive statistics are reported in Appendix~\ref{app:emp_detail_stats}.

\subsection{Model specification}

We estimate two specifications.
In both cases, we use the quadratic perturbation~\eqref{eq:quadF}.

\paragraph{Specification 1 (baseline).}
The baseline link cost function includes link length, GPS-based link travel time, free-flow travel time, and an indicator for junctions in urban areas:
\begin{equation}
Z\beta \,=\, \beta_{{\ell}} \times \text{length}
+ \beta_{\text{tt}}\times \text{time}_{\text{GPS}}
+ \beta_{\text{ff\_tt}}\times \text{time}_{\text{ff}}
+ \beta_{\text{junction}}\times(\text{junction}\times\text{urban}).
\end{equation}

\paragraph{Specification 2 (income heterogeneity).}
To allow for individual heterogeneity in route choice by income, we interact distance with a binary indicator for income $\mathcal{I}\in\{0,1\}$ that equals one if the observation's income is above the Danish population average, and zero otherwise.
The link cost function is specified as
\begin{align}
Z\beta &=\, \beta_{\ell}^{\text{low}}\times\text{length}\times (1-\mathcal{I})
+ \beta_{\ell}^{\text{high}}\times\text{length}\times \mathcal{I} \nonumber \\
&+ \beta_{\text{tt}}\times \text{time}_{\text{GPS}}
+ \beta_{\text{ff\_tt}}\times \text{time}_{\text{ff}}
+ \beta_{\text{junction}}\times(\text{junction}\times\text{urban}).
\end{align}

\subsection{Empirical results}

We estimate both specifications using Algorithm~\ref{iterated_algo}.
For the estimates in this section, we initialize $\tilde \beta_{(0)} = \mathbf{1}$. We show in the next subsection that the estimates are robust to the initialization.

Tables~\ref{tab:parest_baseline} and~\ref{tab:parest_income} report coefficient estimates and standard errors. Our data contains a traveler identifier, hence we compute standard errors both with and without clustering (see Section \ref{sec:repeated_trips}). 
All coefficients are significantly different from zero under the clustered standard errors, which inflate the standard error under the i.i.d. assumption by a factor of 3.5-5.1.
In both specifications, the coefficient on GPS-based link travel time is smaller than the coefficient on the free-flow travel time, consistent with the observation that the free-flow travel time understates the experienced travel time on average.

\begin{table}[H]
    \centering
    \begin{threeparttable}
    \caption{Parameter estimates (baseline specification)}
    \label{tab:parest_baseline}
    \begin{tabular}{lrrrr}
        \toprule
        Attribute & $\tilde\beta$ & Std. error& Clustered & Clustered \\
         &  & (i.i.d.) & std. error & t-stat \\
        \midrule
        Link length & 1.4603 & 0.0469 & 0.2303 & 6.34 \\
        GPS-based link time & 1.8947 & 0.0259 & 0.0912 & 20.77 \\
        Free-flow link time & 2.3932 & 0.0449 & 0.1626 & 14.72 \\
        Link junction $\times$ urban & 0.4691 & 0.0099 & 0.0376 & 12.49 \\
        \bottomrule
    \end{tabular}
    \begin{tablenotes}\footnotesize
      \item Clustered standard errors group the per-trip scores by traveler.
    \end{tablenotes}
    \end{threeparttable}
\end{table}

\begin{table}[H]
    \centering
    \begin{threeparttable}
    \caption{Parameter estimates (income-heterogeneous specification)}
    \label{tab:parest_income}
    \begin{tabular}{lrrrr}
        \toprule
        Attribute & $\tilde\beta$ & Std. error & Clustered & Clustered \\
         &  & (i.i.d.) & std. error & t-stat \\
        \midrule
        Link length $\times (1-\mathcal{I})$ & 1.4576 & 0.0617 & 0.3120 & 4.67 \\
        Link length $\times \mathcal{I}$ & 1.0520 & 0.0591 & 0.2775 & 3.79 \\
        GPS-based link time & 1.8836 & 0.0297 & 0.1096 & 17.18 \\
        Free-flow link time & 2.2704 & 0.0512 & 0.1880 & 12.08 \\
        Link junction $\times$ urban & 0.4806 & 0.0108 & 0.0413 & 11.62 \\
        \bottomrule
    \end{tabular}
    \begin{tablenotes}\footnotesize
      \item Clustered standard errors group the per-trip scores by traveler.
    \end{tablenotes}
    \end{threeparttable}
\end{table}

To interpret magnitudes, we translate the estimated trade-off between GPS-based travel time and distance into an implied value of travel time (VTT) and compare it to census data.
The monetary driving cost is 2.83 DKK per kilometer, according to the Danish Unit Price Catalog \citep{UnitPrices}. This cost includes fuel, tires, repair and maintenance, vehicle tax, and depreciation.
The ratio $\tilde\beta_{tt}/\tilde\beta_{\ell}$ has units of km per minute and therefore translates this driving cost into a VTT.
Specifically,
\begin{equation*}
\text{VTT (DKK/hour)} \,=\, 2.83 \times \frac{\tilde\beta_{tt}}{\tilde\beta_{\ell}} \times 60.
\end{equation*}
Under the baseline specification, this yields
\begin{equation*}
\text{VTT} = 2.83 \times \frac{1.8947}{1.4603} \times 60 \,\approx\, 220\ \text{DKK/hour}.
\end{equation*}

We use reported categorical income collected alongside the experiment and, using the midpoint of the categorical income intervals, estimate the average pre-tax annual wage at about 540,900 DKK for these observations.

Assuming an average tax rate of 40\%, the implied average post-tax annual wage is 324,540 DKK. Using that Danes on average work 1380 hours per year \citep{OECD2024Hours}, this translates into an hourly value of 235 DKK/hour,
which is close to the VTT implied by the data. This indicates that the estimated parameters have a reasonable size.~\footnote{The implied VTT is meant to provide a simple plausibility check on the scale of the estimated coefficients. This comparison is not intended as a formal welfare or policy evaluation, but rather as a rough external check for the magnitudes implied by the model. Differences relative to official guideline values should therefore not be overinterpreted.}

For the income-heterogeneous specification, the implied VTT differs by income group through the length coefficient. Using the same formula with $\tilde\beta_{\ell}^{\text{low}}$ and $\tilde\beta_{\ell}^{\text{high}}$ gives
\begin{align*}
\text{VTT}^{\text{low}} &= 2.83 \times \frac{1.8836}{1.4576} \times 60 \approx 219\ \text{DKK/hour}\\
\text{VTT}^{\text{high}} &= 2.83 \times \frac{1.8836}{1.0520} \times 60 \approx 304\ \text{DKK/hour}.
\end{align*}

The implied VTTs for the low and high (above average) income groups are also consistent with (loosely) computed after-tax hourly wages of 188 DKK/hour and 383 DKK/hour, respectively (by translating the group-specific post-tax annual wage into hourly rates).
If we take into account that probably the low income group works fewer hours and the high income group works more hours~\citep{EU2024HoursByOccu}, then the implied hourly wages come closer to the estimated VTTs. In sum, the implied VTTs are of a plausible magnitude and show a clear and plausible difference across the two income groups in specification~2.

The fitted model also provides insights into the active set, i.e., the endogenous choice set. 
As shown in Figure~\ref{fig:emp-active}, the active set grows with travel distances: longer trips are expected to traverse more of the network, so the perturbed utility route choice model properly captures this and spreads flow over more links.
For a median trip (9.2~km), it activates only about $30$ links. So the predicted active set remains sparse. Notably, even for the longest trips (102~km), the active set contains only 0.35~\% of the $98{,}503$ network links. 

\begin{figure}[H]
  \centering
  \includegraphics[width=0.65\textwidth]{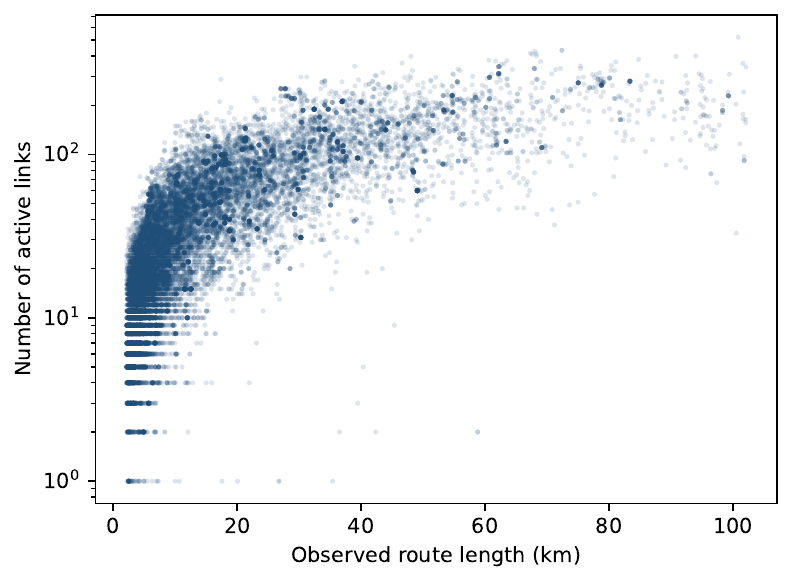}
  \caption{Number of active links versus trip distance.}
  \label{fig:emp-active}
\end{figure}

\paragraph{Computational performance.} Each estimation step solves $N=18{,}019$ trip-level PURC problems on a network of $98{,}503$ directed links. The empirical experiments are conducted on a 64-core workstation with 512~GB memory.
The baseline specification converges in 13 outer iterations in a total of 10.3 minutes, with a peak memory usage of $50.6$~GB; the income-heterogeneous specification is similar. 
Because the $N$ trips are solved independently within each estimation step, the estimator is embarrassingly parallel: we distribute the trips across the available cores, and the per-iteration compute grows linearly in the number of trips while the number of outer iterations stays small, so the method scales to large networks and samples (see Section~\ref{sec:sim} and~\ref{sec:scalability}). The specific hardware, PURC solver, and estimation tuning parameters are collected in Appendix~\ref{app:implementation_details}.

\subsection{Prediction quality}
Beyond the parameter estimates, we assess how well the fitted model reproduces the observed routes, both in-sample and out-of-sample. 
For a trip $n$ with observed link set $\mathcal{R}_n=\{e:y^n_e=1\}$ and predicted flow $\hat x^n=\hat x(Z\tilde\beta,b^n)$, we use two link-level overlap measures between the predicted flow and the observed route,
\begin{equation}
\mathrm{coverage}_n = \frac{\sum_e {1}_{y_e\cdot\hat x^n_e>0}}{|\mathcal{R}_n|},
\qquad
\mathrm{flowshare}_n = \frac{\sum_{e\in\mathcal{R}_n}\hat x^n_e}{\sum_{e\in\Links}\hat x^n_e}.
\end{equation}
Coverage is the fraction of the observed route's links that carry positive predicted flow, and flow share is the fraction of the predicted flow mass that falls on the observed route. We also report the residual variance Eq.~\eqref{eq:RV}, our estimator criterion, evaluated out of sample. 
Finally, at the route level, we report the $R^2$ from regressing the predicted total on the observed total across trips. Because the model predicts a link flow rather than a single route, the predicted total of a link attribute $a$ (route length or travel time) is the flow-weighted expectation $\sum_{e\in\Links}\hat x^n_e\, a_e$, whereas the observed total is $\sum_{e\in\Links} y^n_e\, a_e$.

\paragraph{In-sample fit} Figure~\ref{fig:emp-calibration} plots the predicted expected against the observed route totals for the baseline specification. Both the total length and the total travel time align closely with the $45^\circ$ line, with $R^2=0.99$ and $R^2=0.92$, respectively. 
Aggregating over trips, the predicted flow places $65\%$ of its mass on the observed route (flow share) and covers $81\%$ of the observed route's links (coverage), and the in-sample residual variance is $RV_N(\tilde\beta)=12.83$ (baseline specification). 
Overall, the perturbed utility model yields satisfactory in-sample performance. We report similar performance for the income-heterogeneous specification in Appendix~\ref{app:emp_detail_stats}.

\begin{figure}[H]
  \centering
  \begin{subfigure}{0.48\textwidth}\includegraphics[width=\linewidth]{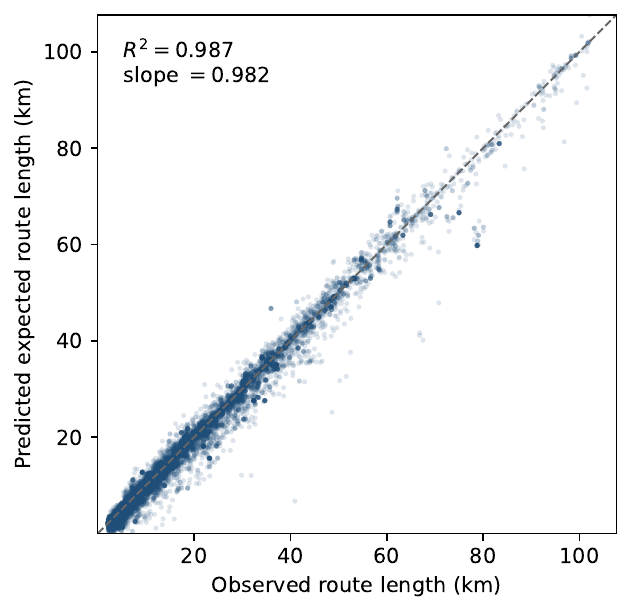}\caption{Route length}\end{subfigure}
  \hfill
  \begin{subfigure}{0.48\textwidth}\includegraphics[width=\linewidth]{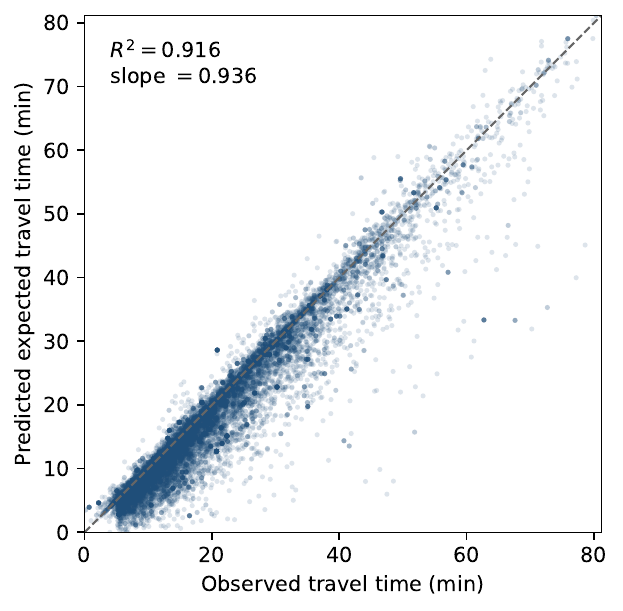}\caption{GPS travel time}\end{subfigure}
  \caption{Predicted versus observed route totals (baseline specification, in sample).}
  \label{fig:emp-calibration}
\end{figure}

\paragraph{Out-of-sample evaluation} We further evaluate the model's out-of-sample performance with 5-fold cross-validation: the model is re-estimated on each training fold (each with $80\%$ of the full dataset) and evaluated on the held-out fold ($20\%$ of the full dataset), and the metrics are computed across folds.

\begin{table}[H]
  \centering
  \caption{Held-out (five-fold) prediction quality. \\Means over folds (standard deviations in parentheses).}
  \label{tab:emp-prediction}
  \begin{tabular}{lccccc}
    \toprule
    Specification & Coverage & Flow share & $RV_N$ & Length $R^2$ & Time $R^2$ \\
    \midrule
    Baseline & 0.807 & 0.650 & 12.83 & 0.987 & 0.916 \\
             & (0.003) & (0.004) & (0.17) & (0.001) & (0.005) \\
    Income-heterogeneous & 0.811 & 0.646 & 13.10 & 0.986 & 0.911 \\
             & (0.004) & (0.003) & (0.08) & (0.001) & (0.004) \\
    \bottomrule
  \end{tabular}
\end{table}

As shown in Table~\ref{tab:emp-prediction}, the out-of-sample values are essentially identical to the in-sample values: coverage remains about $0.81$, flow share of $0.65$, and route-total $R^2$ of $0.99$ (length) and $0.92$ (time). 
In particular, the held-out residual variance $RV_N(\tilde\beta)\approx 12.8$ equals its in-sample value, so the fitted model generalizes to unseen trips in terms of the criterion. 
Similarly, the income-heterogeneous specification performs well out of sample.

\subsection{Robustness to initialization}\label{subsec:emp-init-robust}
Similar to the robustness exercise in the simulation section~\ref{sec:robustness_to_init_sim}, we re-estimate each specification from multiple randomized initial parameter vectors and summarize both (i) the distance from the randomized start to the baseline estimate and (ii) the distance from the final estimate to the baseline estimate.

For each specification, we take the converged estimate from the previous subsection (initialized with $\tilde \beta_{(0)} = \mathbf{1}$) as $\hat\beta$.
We then run $M=10$ replications in which each element of $\tilde \beta_{(0)}$ is drawn independently from $\mathrm{Uniform}[0.5,1.5]$, holding constant all Algorithm~\ref{iterated_algo} parameters (Appendix~\ref{app:implementation_details}).
For each replication $m$, we record the starting distance $\|\tilde \beta_{(0), m}-\hat\beta\|$, the final distance $\|\tilde\beta_m-\hat\beta\|$, the number of outer iterations, and wall-clock runtime.

Table~\ref{tab:emp-init-diagnostics} shows that the random initializations are far from the baseline estimate in Euclidean norm (about 1.8 on average), yet all runs converge.
In the baseline specification, every replication returns essentially the same coefficient vector as the baseline run, as shown in Table~\ref{tab:emp-init-coefs}.
In the income-heterogeneous specification, the final distance to $\hat\beta$ is on the order of $ 10^{-3}$ in every run, and the induced dispersion in coefficients is negligible relative to the point estimates in Table~\ref{tab:parest_income}.
Runtime is about ten to twelve minutes per run
on our 64-core workstation. This demonstrates the scalability of our approach.

\begin{table}[H]
  \centering
  \begin{threeparttable}
  \caption{Robustness to initialization in the empirical application.}
  \label{tab:emp-init-diagnostics}
  \small
  \setlength\tabcolsep{3pt}
  \begin{tabular}{lccccc}
    \toprule
      Specification
      & Starting distance
      & Final distance
      & Convergence
      & Avg.
      & Avg.  \\
      & $\mathbb{E}\big[\|\tilde \beta_{(0)}-\hat\beta\|\big]$
      & $\mathbb{E}\big[\|\tilde\beta-\hat\beta\|\big]$
      & rate
      & iterations
      & runtime (min.) \\
    \midrule
      Baseline
      & 1.703
      & 0.0005
      & 100\%
      & 13.60
      & 12.07 \\

      & (0.320)
      &  (0.0004)
      &
      & (1.35)
      & (3.16) \\
      Income-heterogeneous
      & 1.818
      & 0.0036
      & 100\%
      & 13.40
      & 10.00 \\
      & (0.321)
      & (0.0029)
      &
      & (1.07)
      & (0.91) \\
    \bottomrule
  \end{tabular}
  \begin{tablenotes}
      \footnotesize
      \item Notes: Each row summarizes $M=10$ randomized initializations. Entries report means with standard deviations in parentheses.
    \end{tablenotes}
  \end{threeparttable}
\end{table}

\begin{table}[H]
  \centering
  \begin{threeparttable}
  \caption{Coefficient stability across randomized initializations.}
  \label{tab:emp-init-coefs}
  \small
  \begin{tabular}{lccc}
    \toprule
      Parameter & Baseline $\hat\beta$ & Mean over reps & SD over reps \\
    \midrule
    \multicolumn{4}{l}{\textit{Specification 1 (baseline)}}\\
      $\beta_{\ell}$ (distance) & 1.4603 & 1.4604 & 0.0001 \\
      $\beta_{\text{tt}}$ (GPS time) & 1.8947 & 1.8944 & 0.0003 \\
      $\beta_{\text{ff\_tt}}$ (free-flow time) & 2.3932 & 2.3930 & 0.0004 \\
      $\beta_{\text{junction}}$ (junction$\times$urban) & 0.4691 & 0.4689 & 0.0002 \\
    \midrule
    \multicolumn{4}{l}{\textit{Specification 2 (income-heterogeneous distance)}}\\
      $\beta_{\ell}^{\text{low}}$ (distance, below avg.) & 1.4576 & 1.4573 & 0.0028 \\
      $\beta_{\ell}^{\text{high}}$ (distance, above avg.) & 1.0520 & 1.0528 & 0.0005 \\
      $\beta_{\text{tt}}$ (GPS time) & 1.8836 & 1.8843 & 0.0012 \\
      $\beta_{\text{ff\_tt}}$ (free-flow time) & 2.2704 & 2.2704 & 0.0006 \\
      $\beta_{\text{junction}}$ (junction$\times$urban) & 0.4806 & 0.4811 & 0.0002 \\
    \bottomrule
  \end{tabular}
  \end{threeparttable}
\end{table}

Overall, these results show that our Algorithm~\ref{iterated_algo} remains robust to initialization in the empirical applications, where multiple randomized initial starting points reach the same estimates in our experiments.

\section{Concluding remarks}\label{sec:conclusion}
This paper has introduced microPURC, a micro-level estimator for the perturbed utility route choice (PURC) model. MicroPURC enables consistent, asymptotically normal, and computationally efficient estimation of route choice parameters directly from trip-level data, such as individual GPS traces, without requiring pre-defined choice sets or data aggregation. The estimator is computationally feasible for large networks, as demonstrated in our simulation study and empirical application.

There are several interesting directions for future research. A main direction may be applying the microPURC estimator in new settings, leveraging the strengths of the PURC model and the microPURC estimator with large location datasets that have now become widely available. Another research direction concerns scaling to very large problems, as we can now envisage very large networks and data. It is also of interest to be able to estimate more general PURC models. For example, allowing more flexible specification of the perturbation functions or random parameters. More generally, it is of interest to develop specification tests for the PURC as well as other route choice models. The microPURC estimator solves a bilevel problem. Bilevel estimation problems occur across a range of fields, and future research could extend our methodology to these settings.

\section*{Acknowledgements}
Part of this work has been financed by the European Union - NextGenerationEU. Fosgerau and Nielsen have received funding from the Independent Research Fund Denmark. 

\bibliographystyle{jfe}
\bibliography{MFreferences,madspReferences,nnreferences, TKRAreferences}

\appendix{}

\section{Appendix}\label{appendix}

This appendix comprises three sections.
Section \ref{app:index_notation} provides an index of notation. Section \ref{app:Phat} provides a lemma that establishes the basic properties of the projection matrix $\hat P$ and how it may be efficiently computed. Section \ref{App:asymptotic_theory} provides assumptions, supporting results, and proofs of our main theorems.

\subsection{Index of notation}\label{app:index_notation}

\begin{longtable}{@{}p{2.8cm}@{\hspace{1.2em}}>{\raggedright\arraybackslash}p{\dimexpr\textwidth-2.8cm-1.2em\relax}@{}}
\toprule
Symbol & Meaning \\
\midrule
\endfirsthead
\toprule
Symbol & Meaning \\
\midrule
\endhead
\multicolumn{2}{@{}l@{}}{\underline{Network}} \\
\addlinespace[0.3em]
$\mathcal{N}$ & Set of nodes \\
$\mathcal{L}$ & Set of links \\
$ij$ & Link from node $i$ to node $j$ \\
$v$ & Generic node; $o$ origin, $d$ destination of a trip \\
$A$ & Node--link incidence matrix, entries $a_{v,ij}$ \\
$l_{ij}$ & Length of link $ij$; $L=\mathrm{diag}(l)$ \\
\addlinespace
\multicolumn{2}{@{}l@{}}{\underline{Trip observations and data}} \\
\addlinespace[0.3em]
$(y,Z,b)$ & Generic trip: route indicator, link attributes, OD vector \\
$y\in\{0,1\}^{\mathcal{L}}$ & Link-incidence vector of the observed route \\
$Z$ & Matrix of link attributes; $Z_{ij}$ is the row for link $ij$ \\
$b\in\R^{\mathcal{N}}$ & OD vector: $-1$ at the origin, $+1$ at the destination \\
$c=Z\beta$ & Link cost vector \\
$n$, $N$ & Trip index and number of trips \\
$J$, $j$ & Number of travelers and traveler index (Section~\ref{sec:repeated_trips}) \\
$\mathcal T$ & Trips per traveler (Section~\ref{sec:repeated_trips}) \\
$K$ & Number of parameters \\
$R$ & Number of Monte Carlo replications \\
$\mathcal{I}$ & Indicator for above-average income (Specification 2) \\
\addlinespace
\multicolumn{2}{@{}l@{}}{\underline{Model primitives}} \\
\addlinespace[0.3em]
$x$ & Link flow vector, $x\in\R_+^{\mathcal{L}}$ \\
$\hat x(c,b)$ & Optimal PURC flow given cost vector $c$ and OD vector $b$ \\
$x_0=\hat x(Z\beta_0,b)$ & Optimal flow at the true parameter \\
$x^*=\hat x(Z\beta^*,b)$ & Optimal flow at a fixed point $\beta^*$ of $\Phi$ \\
$F$, $F_{ij}$ & Perturbation function and its link components \\
$\nabla F$, $\nabla^2F$ & Gradient and Hessian of $F$ \\
$\Lambda$ & Lagrangian of the perturbed cost minimization problem \\
$\hat\eta$ & Multipliers on flow conservation (node potentials) \\
$\hat\mu$ & Multipliers on the non-negativity constraint $x\ge 0$ \\
$s_{ij}$ & Dual slack of link $ij$; $ij$ active iff $s_{ij}<0$ \\
$\beta$, $\beta_0$ & Cost parameter vector and its true value \\
\addlinespace
\multicolumn{2}{@{}l@{}}{\underline{Active set and projection}} \\
\addlinespace[0.3em]
$\hat B$ & $\mathrm{diag}(1_{\hat x>0})$, indicator of active links \\
$\hat P$ & Orthogonal projection onto circulations on the active subnetwork \\
$P^*$ & $\hat P$ evaluated at $x^*$ \\
$\tilde P^n$, $\tilde P^n_*$ & Projection for trip $n$ in reduced (active-block) form \\
$(\cdot)^+$ & Moore--Penrose pseudoinverse \\
$m_n$ & Active-set margin of trip $n$, in $\beta$-units \\
\addlinespace
\multicolumn{2}{@{}l@{}}{\underline{Estimation}} \\
\addlinespace[0.3em]
$\Ebb$, $\Ebb_N$ & Population expectation and empirical average over trips \\
$RV$, $RV_N$ & Residual variance criterion and its sample analogue \\
$T(\beta,x,Z,y)$ & Bias-corrected projected first-order condition \\
$\Phi$, $\Phi_N$ & Iterative map and its sample analogue \\
$\tilde\beta$ & The microPURC estimator \\
$\bar\beta$ & Auxiliary minimizer of the merit function (proof of Thm.~\ref{thm:normality}) \\
$\hat g^n$ & Score of trip $n$ \\
$H$, $S$ & Curvature and score-variance matrices of the sandwich form \\
$\tilde H$, $\tilde S$ & Sample counterparts of $H$ and $S$; $\tilde H$ is also the matrix inverted in Step 2 of Algorithm~\ref{iterated_algo} \\
$\tilde S_{\mathrm{cl}}$ & Traveler-clustered counterpart of $\tilde S$ \\
$\tilde\sigma$, $\tilde\sigma_{\mathrm{cl}}$ & Standard errors, i.i.d.\ and traveler-clustered \\
\addlinespace
\multicolumn{2}{@{}l@{}}{\underline{Algorithm~\ref{iterated_algo}}} \\
\addlinespace[0.3em]
$\tilde\beta_{(t)}$ & Parameter iterate at outer iteration $t$; $\tilde\beta_{(0)}$ the starting value \\
$Q_N$ & Merit function $\tfrac12\|\beta-\Phi_N(\beta)\|^2$ \\
$\alpha_{(t)}$ & Krasnoselskii--Mann step size \\
$\gamma$ & Backtracking reduction factor \\
$M$ & Width of the non-monotone line-search window \\
$m_{\max}$ & Maximum number of line-search trials \\
$\tau$ & Fixed-point residual tolerance \\
\addlinespace
\multicolumn{2}{@{}l@{}}{\underline{Asymptotics}} \\
\addlinespace[0.3em]
$\mathcal{B}_0$ & Neighborhood of $\beta_0$ with constant active sets (Ass.~\ref{as:regularity}(e)) \\
$\mathcal{B}$, $B_r$ & Closed balls centered at $\beta_0$, the latter of radius $r$ \\
$\epsilon_N$ & Vanishing linear coefficient in the convergence bound \\
$\parrow$, $\Darrow$ & Convergence in probability and in distribution \\
$\norm{\cdot}$ & Euclidean norm; operator norm for matrices \\
\bottomrule
\end{longtable}

\subsection{The projection matrix}\label{app:Phat}

The following lemma provides the basic properties of the matrix $\hat P$. By items b) and d), $\hat P$ is a symmetric and idempotent matrix, hence it is an orthogonal projection matrix. By items a) and c), it is a projection on the linear subspace \eqref{eq:subspace}, defined in the main text.

\begin{lemma}\label{lem:P_hat}
Define
\begin{equation*}
\hat{P}=\hat{B}-(A \hat B)^+A\hat B,
\end{equation*}%
Then $\hat{P}$ is an orthogonal projection matrix on the linear subspace $\{x\in \R^{\mathcal{L}}:Ax=0, \hat B x=x\}$. Moreover,
\begin{itemize}
    \item[a)] $\hat{P}=\hat{P}\hat{B}$;
    \item[b)] $\hat P = \hat P ^\top = \hat B - \hat B A^\top \left(\hat B     A^\top\right)^+$.
    \item[c)]$\hat{P}A^{\top }=0$;
    \item[d)] $\hat P^2 = \hat P$ ;
\end{itemize}
\end{lemma}

\begin{proof}                                                                  \begin{itemize}
    \item[a)] It follows by the definition of $\hat P$ that $\hat P \hat B = \hat P $.
    \item[b)]  
    Note that
    \[
    \hat P^\top = \hat B^\top - \left[(A \hat B)^+A\hat B\right]^\top =  \hat B - (A \hat B)^+A\hat B = \hat P,
    \]
    where we used $\hat B^\top=\hat B$ and the Moore-Penrose inverse property $\left[(A \hat B)^+A\hat B\right]^\top = (A \hat B)^+A\hat B$. Then, it follows that
    \[
     \hat P^\top =\hat B^\top - (A\hat B)^\top\bigl((A\hat B)^+\bigr)^\top
    =\hat B - \hat B A^\top \bigl((A\hat B)^\top\bigr)^+
    =\hat B - \hat B A^\top (\hat B A^\top)^+,
    \]
    using that $(A\hat B)^\top=\hat B A^\top$ and the property that Moore-Penrose inverses commute with transposition: $((A\hat B)^+\bigr)^\top = \bigl((A\hat B)^\top\bigr)^+$.
    \item[c)]Note that
    \begin{eqnarray*}
        \hat{P}A^{\top } &=& \hat{P}^\top \hat B A^{\top } \\
        &=&    \left( \hat B - \hat B A^\top \left(\hat B     A^\top\right)^+\right) \hat B A^{\top }   \\
        &=& \hat B A^{\top } - \hat B A^\top \left(\hat B     A^\top\right)^+ \hat B A^{\top } =0,
    \end{eqnarray*}
    where the last line follows from the Moore–Penrose condition.
    \item[d)] Follows by straightforward computation.
\end{itemize}

To show that $\hat P$ is an orthogonal projection on the linear subspace $S=\{x\in \R^{\mathcal{L}}:Ax=0, \hat B x=x\}$, note first that we have shown that $\hat P$ is idempotent and symmetric. It remains to show that $\mathrm{range}(\hat P)=S$. For $y\in S$ we note that $\hat P y = y$, such that $S \subseteq \mathrm{range}(\hat P)$. Conversely, for any $z=\hat Py$, we have i) $\hat Bz = \hat B \hat Py = \hat Py = z$, and ii) $Az = A\hat B \hat P y = A\hat By - A\hat B (A\hat B)^+A\hat B y =0$. Hence, $\mathrm{range}(\hat P) \subseteq S$.

\end{proof}
\bigskip

The theory section below uses matrices and vectors that have the size of the full network. This is convenient for the theory, but not for computation.
For illustration purposes, label the links such that all active links come first, and remove the row corresponding to the destination in order for the node-link incidence matrix to have full rank. The projection matrix then has the form
\begin{eqnarray}\label{eq:Ptilde}
\tilde P^n=%
\begin{pmatrix}
\tilde P^n_* & 0 \\
0 & 0%
\end{pmatrix}%
,
\end{eqnarray}
Let $\tilde A^n$ be the block of the adjacency matrix corresponding to the
active network (except for having removed the destination). Then $\tilde P^n_*
= I-\tilde A^{n^{\top }}(\tilde A^n \tilde A^{n^{\top }})^{-1}\tilde A^n$%
. Letting $\tilde Z^n$ denote the relevant columns of $Z^n$ and $\tilde
l^n\circ \tilde y^n$ denote the relevant columns of $l\circ y^n$, we obtain
\begin{eqnarray}
(Z^n\beta+l\circ y^n)^{\top }\tilde P^n (Z^n\beta +l\circ y^n) =(\tilde
Z^n\beta +\tilde l^n \circ \tilde y^n)^{\top }\tilde{P}_*^n (\tilde
Z^n\beta+ \tilde l^n \circ \tilde y^n).
\end{eqnarray}
The same principle applies throughout.

\subsection{Asymptotic theory}\label{App:asymptotic_theory}

To establish context for the next lemma, recall that the map $\Phi$ given in eq. \eqref{eq:Phi} minimizes the Euclidean norm of the approximate projected first-order condition  $T$. We now establish some mathematical properties of this map.

\begin{lemma}[Properties of $\Phi$,  $\Phi_N$ and $\hat x$]\label{lem:Phi_properties}
 Let $\mathcal{B}\subset \mathcal{B}_0$ be a closed ball centered at $\beta_0$, where $\mathcal{B}_0$ is given by Assumption \ref{as:regularity}(e). Then we have
\begin{itemize}
    \item $\nabla \Phi(\beta_0)=0$;
    \item $\Phi(\beta)\neq \beta$ for all $\beta\neq \beta_0$ in a neighborhood of $\beta_0$;
    \item  $\sup_{\beta\in \mathcal{B}}\norm{\nabla^2 \Phi(\beta)}<\infty$;
    \item $\sup_{\beta\in \mathcal{B}}\norm{\Phi_N(\beta)-\Phi(\beta)}\parrow 0$,  $\sup_{\beta\in \mathcal{B}}\norm{\nabla \Phi_N(\beta)-\nabla \Phi(\beta)}\parrow 0$, \\ and $\sup_{\beta\in \mathcal{B}}\norm{\nabla^2 \Phi_N(\beta)-\nabla^2 \Phi(\beta)}\parrow 0$.
    \item For almost all $Z,b$, $\hat x(Z\beta,b)$ is three times continuously differentiable as a function of $\beta$ for all $\beta\in \mathcal{B}_0$.
\end{itemize}
\end{lemma}
\begin{proof}
We begin by inspecting the function $T$ given in \eqref{eq:T_function},
\begin{eqnarray*}
    T(\beta,x,Z,y)=\hat P(x)(Z\beta+\nabla F(x)+\nabla^2 F(x)(y-x))
\end{eqnarray*}
Since $\hat P$ is a nonzero orthogonal projection by Lemma \ref{lem:P_hat}, we have $\norm{\hat P(x)}= 1$ where for a matrix, $\norm{\cdot}$ denotes the Euclidean operator norm i.e. the maximal singular value. Applying the triangle inequality and submultiplicativity, we find
\begin{eqnarray*}
    \norm{T(\beta,x,Z,y)}\leq \norm{Z}\norm{\beta} +\norm{\nabla F(x)}+\norm{\nabla^2 F(x)}\norm{y-x}.
\end{eqnarray*}
by continuity of $\nabla F$ and $\nabla^2 F$, we find
\begin{eqnarray*}
    \sup_{\beta\in \mathcal{B},x\in [0,1]^{\mathcal{L}},y\in \{0,1\}^{\mathcal{L}}}\norm{T(\beta,x,Z,y)}\leq C_1\norm{Z}+C_2,
\end{eqnarray*}
where $C_1=\sup_{\beta\in \mathcal{B}}\norm{\beta}$ and $C_2=\sup_{x\in [0,1]^{\mathcal{L}}}\norm{\nabla F(x)} + \sup_{x\in [0,1]^{\mathcal{L}},y\in \{0,1\}^{\mathcal{L}}}\left\{\norm{\nabla^2 F(x)}\norm{y-x}\right\}$ are both finite by continuity over compact sets. Then
\begin{eqnarray*}
    \sup_{\beta\in \mathcal{B},x\in [0,1]^{\mathcal{L}},y\in \{0,1\}^{\mathcal{L}}}\norm{T(\beta,x,Z,y)}^2\leq C_1^2 \norm{Z}^2+2C_1 C_2 \norm{Z}+C_2^2,
\end{eqnarray*}
such that $ \sup_{\beta\in \mathcal{B},\beta'\in \mathcal{B}}\Ebb \left[\norm{T(\beta,\hat x(Z\beta',b),Z,y)}^2\right]<\infty $ by Assumption \ref{as:regularity}(d).

By Assumption \ref{as:regularity}(e), the set of active links is fixed for $\beta\in \mathcal{B}\subset \mathcal{B}_0$, such that $\hat P(\hat x(Z\beta,b))=\hat P(\hat x(Z\beta_0,b))\equiv \hat P$ for all $\beta\in \mathcal{B}$. By the implicit function theorem \citep[][Thm. 3.3.1]{krantz2012implicit}, $\hat x(Z\beta,b)$ is three times continuously differentiable as a function of $\beta$ in a neighborhood of $\beta_0$, for almost all $Z,b$ (see the proof of Theorem 1 in \citet{fosgerau_sensitivity_2025} for the precise argument).    Then
\begin{eqnarray*}
    & & \frac{\partial}{\partial \beta_k}\left(\nabla F(\hat x)+\nabla^2 F(\hat x)(y-\hat x)\right) \\
    &=& \nabla^2 F(\hat x)\frac{\partial \hat x}{\partial \beta_k}-\nabla^2 F(\hat x)\frac{\partial \hat x}{\partial \beta_k} +\frac{\partial \nabla^2 F(\hat x)}{\partial \beta_k}(y-\hat x) \\
    &=& \frac{\partial \nabla^2 F(\hat x)}{\partial \beta_k}(y-\hat x) .
\end{eqnarray*}
From \eqref{eq:Phi_closedform} we therefore find that for $\beta\in \mathcal{B}$,
\begin{eqnarray*}
\frac{\partial \Phi(\beta)}{\partial \beta_k} &=&-\left(\Ebb\left[ Z^\top \hat P Z\right]\right)^{-1} \Ebb\left[ Z^\top \hat P\frac{\partial \nabla^2 F(\hat x(Z\beta,b))}{\partial \beta_k}(y-\hat x(Z\beta,b)) \right] \\
&=& -\left(\Ebb\left[ Z^\top \hat P Z\right]\right)^{-1} \Ebb\left[ Z^\top \hat P\frac{\partial \nabla^2 F(\hat x(Z\beta,b))}{\partial \beta_k}(\hat x(Z\beta_0,b)-\hat x(Z\beta,b)) \right],
\end{eqnarray*}
where the second equality applies the law of iterated expectations and Assumption \ref{as:regularity}(b). It follows that $\nabla \Phi(\beta_0)=0$.

Defining $h(\beta)=\beta-\Phi(\beta)$ we have that $\nabla h(\beta_0)=I$. It now follows by the implicit function theorem \citep[Theorem 3.3.2 in][]{krantz2012implicit} that $h$ is one-to-one on an open set containing $\beta_0$, such that $h(\beta)=0$ if and only if $\beta=\beta_0$ on this set.

Inspecting the expression for $\frac{\partial \Phi(\beta)}{\partial \beta_k}$, we observe that $\nabla^2 \Phi(\beta)$ exists and is continuous by Assumption \ref{as:regularity}(c). It follows from compactness of $\mathcal{B}$ that  $\sup_{\beta\in \mathcal{B}}\norm{\nabla^2 \Phi(\beta)}<\infty$. Since $\Phi_N$ ($\nabla \Phi_N$ and $\nabla^2 \Phi_N$) is the sample average of $\Phi$ ($\nabla \Phi, \nabla^2 \Phi$) and $\mathcal{B}$ is compact, the last point follows from the uniform law of large numbers.
\end{proof}
\bigskip

The rest of this section provides proofs of the results stated in the main text.

\begin{proof}[Proof of Lemma \ref{lem:unique_RV}]
Let $\epsilon = y - \Ebb[y | Z, b]$ so that $\Ebb[\epsilon|Z, b] = 0$. For any $\beta \in \mathbb{R}^K$, we have
\begin{align*}
    RV(\beta) &= \Ebb\left[\norm{y - \hat{x}(Z\beta, b)}^2\right] = \Ebb\left[\norm{\Ebb[y | Z, b] + \epsilon - \hat{x}(Z\beta, b)}^2\right] \\
    &= \Ebb\left[\norm{\Ebb[y | Z, b] - \hat{x}(Z\beta, b)}^2\right] + 2 \Ebb \left[ \epsilon^\top (\Ebb[y | Z, b] - \hat{x}(Z\beta, b)) \right] + \Ebb\left[\norm{\epsilon}^2\right].
\end{align*}
Using the law of iterated expectations, we have
\begin{equation*}
    \Ebb \left[ \epsilon^\top (\Ebb[y | Z, b] - \hat{x}(Z\beta, b)) \right] = \Ebb \left[ \Ebb[\epsilon|Z, b]^\top (\Ebb[y | Z, b] - \hat{x}(Z\beta, b)) \right] = 0,
\end{equation*}
so that
\begin{equation*}
    RV(\beta) = \Ebb\left[\norm{\Ebb[y | Z, b] - \hat{x}(Z\beta, b)}^2\right] + \Ebb\left[\norm{\epsilon}^2\right].
\end{equation*}

By Assumption~\ref{as:regularity}(b) we have that $\Ebb[y | Z, b] = \hat{x}(Z\beta_0, b)$ so $RV(\beta_0) = \Ebb\left[\norm{\epsilon}^2\right]$.
Further, subtracting $RV(\beta_0)$ from $RV(\beta)$ yields
\begin{eqnarray}\label{eq:residual_inequality}
    RV(\beta)-RV(\beta_0)= \Ebb\left[\norm{\Ebb[y | Z, b] - \hat{x}(Z\beta, b)}^2\right] \geq 0.
\end{eqnarray}
So, we have $RV(\beta) \geq RV(\beta_0), \forall \beta \in \mathbb{R}^K$.
When $RV(\beta) = RV(\beta_0)$, this implies $\Ebb[y | Z, b] = \hat{x}(Z\beta, b)$ almost surely. By the uniqueness in Assumption \ref{as:regularity}(b), this implies $\beta = \beta_0$, and consequently that $\beta_0$ is the unique minimizer of $RV$.
\end{proof}\bigskip

\begin{proof}[Proof of Lemma \ref{lem:fxp_foc}]
The first-order condition for the minimization problem in \eqref{eq:Phi} is
\begin{eqnarray*}
  0 &=&\Ebb\left[ Z^\top P^* (Z\beta^*+\nabla F(x^*)+\nabla^2 F(x^*)(y-x^*))\right]\\
    &=& \Ebb\left[ Z^\top P^*\nabla^2 F(x^*)(y-x^*)\right]\\
    &=&\Ebb\left[ Z^\top P^*\nabla^2 F(x^*)(\Ebb[y|Z,b]-x^*)\right]\\
    &=& \Ebb\left[ Z^\top P^*\nabla^2 F(x^*)(\hat x(Z\beta_0,b)-x^*)\right]
\end{eqnarray*}
where the first simplification follows from the PUM first-order condition  \eqref{eq:projectedFOC} and the second by iterated expectations and Assumption \ref{as:regularity}(b), noting that the law of iterated expectations applies since
\begin{eqnarray*}
 & &  \Ebb[\norm{Z^\top P^*\nabla^2 F(x^*)(y-x^*)}] \\
&<&     \Ebb[\norm{Z^\top }] \sup_{y,x\in [0,1]^{\mathcal{|L|}}}  \norm{\hat P(x)\nabla^2 F(x)(y-x)} <\infty
\end{eqnarray*}
where $\Ebb[\norm{Z^\top }]<\infty$ by Assumption \ref{as:regularity}(d) and $\sup_{y,x\in [0,1]^{|\mathcal{L}|}}  \norm{\hat P(x)\nabla^2 F(x)(y-x)}<\infty$ by continuity over a compact set.
\end{proof}

\bigskip

\begin{proof}[Proof of Lemma \ref{lem:quadratic_convergence}]
    A second-order mean value inequality in $\Phi$ (\citet{zeidler2012applied} Theorem 4.C) yields
\begin{eqnarray*}
   \norm{\Phi(\beta)-\beta_0}&=& \norm{\Phi(\beta)-\Phi(\beta_0)} \\
   &\leq  & \norm{\nabla \Phi(\beta_0)}\norm{\beta-\beta_0}+\sup_{\tau \in [0,1]}\norm{\nabla^2 \Phi(\tau \beta+(1-\tau)\beta_0)}\norm{\beta-\beta_0}^2 \\
   &\leq & C\norm{\beta-\beta_0}^2,
\end{eqnarray*}
given that $\beta$ is sufficiently close to $\beta_0$ by Lemma \ref{lem:Phi_properties}.
\end{proof}\bigskip

\begin{proof}[Proof of Theorem \ref{thm:quadratic_convergence}]
By Lemma~\ref{lem:Phi_properties} and Assumption~\ref{as:regularity}, we have the following:
\begin{enumerate}
    \item[(a)] For almost all $Z,b$, $\hat x(Z\beta,b)$ is three times continuously differentiable as a function of $\beta$ for all $\beta\in \mathcal{B}_0$.  This implies that $\Phi$ and $\Phi_N$ are twice continuously differentiable for all $\beta \in \mathcal B_0$;
    \item[(b)] There exists a constant $C_H > 0$, such that $C_H=\sup_{\beta\in \mathcal{B}_0}\norm{\nabla^2 \Phi(\beta)} < \infty$;
    \item[(c)] The map $\Phi$ satisfies $\nabla\Phi(\beta_0)=0$, and  $\Phi_N\to\Phi$ and $\nabla\Phi_N\to\nabla\Phi$ uniformly on $\mathcal B_0$.
\end{enumerate}

Let $B_r\subset\mathcal B_0$ be a closed ball with radius $r$ centered at $\beta_0$. By item (a), $\Phi_N$ is twice continuously differentiable on $B_r$. We will show that $\Phi_N$ is a contraction mapping on $B_r$.

By $\nabla\Phi(\beta_0)=0$ in item (c) and the continuity of $\Phi$ and $\nabla\Phi$ in item (a), we can pick $r$ such that
\[
\sup_{\beta\in B_r} \|\nabla \Phi(\beta)\| < 1,  \quad \text{and} \quad \sup_{\beta\in B_r} \|\Phi(\beta)- \beta_0\| < r.
\]
To see this, fix $C_L < 1$, and note that by continuity of $\nabla \Phi$ around $\beta_0$ and since $\nabla\Phi(\beta_0)=0$, there exists $r'>0$ such that
\[
\sup_{\beta \in B_{r'}} \|\nabla \Phi(\beta)\| < C_L,
\]
where $B_{r'}$ is a closed ball centered at $\beta_0$ with radius $r'$. Applying the mean value inequality~\citep[][Thm 9.19]{rudin1976principles}, for any $\beta \in B_{r'}$, we have
\begin{align*}
    \norm{\Phi(\beta) -  {\beta_0}}  = \norm{\Phi(\beta) -  \Phi({\beta_0})} \leq C_L  \norm{\beta -\beta_0}  \leq C_L  r'.
\end{align*}
Now, just pick $r = r'$ and take the $\sup$ over $B_r$. Then $ \sup_{\beta\in B_r} \|\Phi(\beta)- \beta_0\| \leq C_L r<r$.

For large $N$, we have by uniform convergence in item (c) that,
\begin{equation}\label{eq:cond_closed_ball}
    \sup_{\beta\in B_r} \|\nabla \Phi_N(\beta)\| < 1,  \quad \text{and} \quad \sup_{\beta\in B_r} \|\Phi_N(\beta)- \beta_0\| < r,
 \end{equation}
with probability approaching 1. In particular, the second inequality implies that $\Phi_N$ is a self-map on $B_r$, that is $\Phi_N(B_r) \subseteq B_r$.

We are now ready to show $\Phi_N: B_r \rightarrow B_r$ is a contraction mapping.
For $\beta, \beta' \in B_r$, we have by first-order Taylor expansion that
\begin{align*}
    \norm{\Phi_N(\beta)-\Phi_N(\beta')} \leq \sup_{\xi\in B_r} \|\nabla \Phi_N(\xi)\| \norm{\beta-\beta'}.
\end{align*}
By construction of $B_r$, we have $\sup_{\xi\in B_r} \|\nabla \Phi_N(\xi)\| < 1$. Hence, $\Phi_N$ is a contraction mapping on $B_r$. Existence and uniqueness of the fixed point follow by a direct application of the Banach fixed point theorem.

Let $\tilde{\beta}$ denote the unique fixed point satisfying $\tilde{\beta} = \Phi_N(\tilde{\beta})$ on $B_r$.
By item (c) and the continuity of $\nabla\Phi$, the consistency $\tilde\beta  \overset{p}{\to}\beta_0$ (Theorem~\ref{thm:normality}) and $\nabla\Phi(\beta_0)= 0$ imply $\norm{\nabla\Phi(\tilde\beta)}\overset{p}{\to} 0$. By uniform convergence in item (c), we have that
\[
  \epsilon_N \coloneqq \norm{\nabla\Phi_N(\tilde\beta)}\overset{p}{\to} 0.
\]

We now establish convergence of Algorithm~\ref{iterated_algo}.
For any $\beta\in B_r$, let $\delta =\beta - \tilde \beta$, since $\tilde\beta\in B_r$ with probability approaching 1, we can write
\[
\Phi_N(\beta) = \Phi_N(\tilde \beta) + \nabla \Phi_N(\tilde \beta) \delta + R(\delta),
\]
where $R(\delta)$ denotes the remainder function. Specifically, let $\xi$ lie on the line segment between $\beta$ and $\tilde\beta$. By item (b), we have that $R$ satisfies,
\[
  \|R(\delta)\|
  \;\le\;
  \tfrac12 \sup_{\xi\in B_r} \bigl\|\nabla^2\Phi_N(\xi)\bigr\|
  \,\|\delta\|^2
  \;\le\;
  C_R \,\|\delta\|^2
\]
for some constant $C_R>0$ independent of $\beta\in B_r$.
Since $\Phi_N(\tilde\beta)=\tilde\beta$ is the fixed point, this yields, for all
$\tilde\beta_{(t)}\in B_r$,
\begin{equation}\label{eq:taylor-phi}
  \Phi_N(\tilde\beta_{(t)}) - \tilde\beta =   \nabla\Phi_N(\tilde\beta)\,\delta_t + R_t,
  \qquad \|R_t\| \le C_R\|\delta_t\|^2.
\end{equation}

The iteration of Algorithm~\ref{iterated_algo} can be written as
\[
  \tilde\beta_{(t+1)}
  =
  (1-\alpha_{(t)})\tilde\beta_{(t)}
  +
  \alpha_{(t)}\Phi_N(\tilde\beta_{(t)}),
  \qquad \alpha_{(t)}\in(0,1].
\]
Subtracting $\tilde\beta$ from both sides gives
\[
  \delta_{t+1}
  =
  (1-\alpha_{(t)})\delta_t
  +
  \alpha_{(t)}\bigl(\Phi_N(\tilde\beta_{(t)})-\tilde\beta\bigr).
\]
Substituting \eqref{eq:taylor-phi} yields
\[
  \delta_{t+1}
  =
  \bigl[(1-\alpha_{(t)})I+\alpha_{(t)}\nabla\Phi_N(\tilde\beta)\bigr]\delta_t
  + \alpha_{(t)} R_t.
\]
Taking norms and using $\alpha_{(t)}\in(0,1]$ and
$\|\nabla\Phi_N(\tilde\beta)\|=\epsilon_N$,
\begin{align*}
  \|\delta_{t+1}\|
  &\le
  \bigl(1-\alpha_{(t)}+\alpha_{(t)}\|\nabla\Phi_N(\tilde\beta)\|\bigr)\|\delta_t\|
  + \alpha_{(t)}\|R_t\| \\
  &\le
  \bigl(1-\alpha_{(t)}+\alpha_{(t)}\epsilon_N\bigr)\|\delta_t\|
  + C_R\|\delta_t\|^2.
\end{align*}

Next, we show that, on $B_r$, $\alpha_{(t)} = 1, \forall t \geq 0$. 
Write $g(\beta)\coloneqq \Phi_N(\beta)-\beta$ for the gap and set $\eta \coloneqq \sup_{\xi\in B_r}\norm{\nabla \Phi_N(\xi)}$, which is strictly below one by Eq.~\eqref{eq:cond_closed_ball}. Hence, $\Phi_N$ is an $\eta$-contraction uniformly on the closed ball $B_r$.

Suppose $\tilde\beta_{(t)}\in B_r$. 
The trial point of the full step $\alpha_{(t)} = 1$ is $\Phi_N(\tilde\beta_{(t)})$, which lies in $B_r$ by the self-map property, so applying the contraction to the pair $\Phi_N(\tilde\beta_{(t)}),\tilde\beta_{(t)}$ gives
\begin{align*}    \norm{g\bigl(\Phi_N(\tilde\beta_{(t)})\bigr)}
    &= \norm{\Phi_N\bigl(\Phi_N(\tilde\beta_{(t)})\bigr) - \Phi_N(\tilde\beta_{(t)})} \\
    &\leq \eta \norm{\Phi_N(\tilde\beta_{(t)}) - \tilde\beta_{(t)}}
    \;=\; \eta \norm{g(\tilde\beta_{(t)})}.
\end{align*}
Recall the merit function
\[
Q_N(\beta) = \frac{1}{2}\norm{g(\beta)}^2,
\]
which is a strictly increasing function of $\norm{g(\beta)}$. Hence, $Q_N(\Phi_N(\tilde\beta_{(t)})) \leq \eta^2 Q_N(\tilde \beta_{(t)}) \leq \max_{0\le i\le \min(t,M)} Q_N(\tilde\beta_{(t-i)})$,
so the non-monotone condition is met at $m=0$ and $\alpha_{(t)}=1$ is accepted, giving $\tilde\beta_{(t+1)}=\Phi_N(\tilde\beta_{(t)})\in B_r$. 
Since $\tilde\beta_{(0)}\in B_r$ by assumption, induction gives
$\alpha_{(t)} = 1, \forall t \geq 0$, and we have
\begin{equation}\label{eq:quadratic_bound}
    \|\delta_{t+1}\|
  \;\le\;
  \epsilon_N\|\delta_t\| + C_R\|\delta_t\|^2.
\end{equation}
This is exactly the bound \eqref{eq:aq-bound-app} with $C=C_R$.

Since $\epsilon_N \xrightarrow{p }0$, the inequality above gives \eqref{eq:aq-bound-app}, and the resulting convergence is almost quadratic in the sense that the linear coefficient
$\epsilon_N$ vanishes in probability as $N\to\infty$.
\end{proof}\bigskip

\begin{proof}[Proof of Theorem \ref{thm:normality}] Consider the estimator for an auxiliary variable $\bar \beta$
\begin{eqnarray}
    \bar \beta=\argmin_{\beta\in \mathcal{B}}Q_N(\beta), \qquad Q_N(\beta)=\frac{1}{2}\left\|\beta-\Phi_N(\beta)\right\|^2,
\end{eqnarray}
where $\mathcal{B}\subset \mathcal{B}_0$ is a sufficiently small closed ball centered at $\beta_0$ such that $\beta_0$ is the only fixed point of $\Phi$ in $\mathcal{B}$. Existence of such $\mathcal{B}$ is ensured by Lemma \ref{lem:Phi_properties}, and $\mathcal{B}_0$ was defined in Assumption \ref{as:regularity}(e).

$Q_N$ converges uniformly in probability to $Q(\beta)=\frac{1}{2}\left\|\beta-\Phi(\beta)\right\|^2$ on $\mathcal{B}$ by the uniform law of large numbers. Note that $Q(\beta)=0$ if and only if $\beta=\beta_0$ by Lemma \ref{lem:Phi_properties}. Therefore, $\bar\beta\parrow \beta_0$ by \citet[][Thm. 2.1]{newey_large_1994} .

The proof proceeds as follows: We will prove the claims of Theorem \ref{thm:normality} for the estimator $\bar \beta$ rather than $\tilde \beta$, and then finally show that $\tilde \beta=\bar \beta$ with probability approaching $1$. We will now verify the conditions of \citet[][Thm. 3.1]{newey_large_1994} for $\bar \beta$.

\textit{Condition i)} $\beta_0\in \mathrm{int}(\mathcal{B})$ by construction. \textit{Condition ii)} With probability 1, $ Q_N(\beta)$ is twice continuously differentiable throughout $\mathcal{B}$ since $\mathcal{B}\subset \mathcal{B}_0$.  \textit{Condition iii)} Since $\Ebb\left[ Z^\top \hat P Z \right]$ is invertible by Assumption \ref{as:regularity}(d), $\Ebb_N\left[ Z^{n,\top }\hat P^nZ^n\right]$ is invertible with probability approaching $1$. In that case, $\Phi_N$ is given in closed form: Define  $\beta\in \mathcal{B}_0$ with $\hat x^n =\hat x(Z^n\beta,b^n)$ and $\hat P^n=\hat P(\hat x^n)$, then we have
\begin{eqnarray*}
    \Phi_N(\beta) &=& -\left( \mathbb E_N(Z^{n,\top}\hat P^n Z^n)\right)^{-1} \Ebb_N\left[ Z^{n,\top}\hat P^n\left(\nabla F(\hat x^n)+\nabla^2 F(\hat x^n)(y^n-\hat x^n) \right)\right]\\
    &=& -\left( \mathbb E_N(Z^{n,\top}\hat P^n Z^n)\right)^{-1} \Ebb_N\left[ Z^{n,\top}\hat P^n\left(-Z^n\beta+\nabla^2 F(\hat x^n)(y^n-\hat x^n) \right)\right] \\
    &=&\beta -\left( \mathbb E_N(Z^{n,\top} \hat P^n Z^n)\right)^{-1} \Ebb_N\left[ Z^{n,\top} \hat P^n\nabla^2 F(\hat x^n)(y^n-\hat x^n)\right],
\end{eqnarray*}
where the second line applies \eqref{eq:projectedFOC} while the third line uses that $\hat P^n$ is constant whenever $\beta\in \mathcal{B}_0$ by Assumption \ref{as:regularity}(e).

We then have
\begin{eqnarray*}
    \sqrt{N}\nabla Q_N(\beta_0)&=& \sqrt{N}(I-\nabla \Phi_N(\beta_0))^\top (\beta_0-\Phi_N(\beta_0))\\
    &=& \left(I-\nabla \Phi_N(\beta_0)\right)^\top \left( \mathbb E_N\left[ Z^{n,\top} \hat P^n Z^{n}\right]\right)^{-1}\sqrt{N}\left(\Ebb_N\left[Z^{n,\top}\hat  P^n \nabla^2 F(\hat x^n)(y^n-\hat x^n) \right]\right).
\end{eqnarray*}
By the law of large numbers, and $\nabla \Phi(\beta_0)=0$, the first term on the RHS converges in probability to $I$ while the second term converges in probability to $H^{-1}=\left(\mathbb E\left[Z^\top \hat PZ\right]\right)^{-1}$. The third term converges in distribution to $\mathbf{N}(0,S)$ by the classical central limit theorem where
\begin{eqnarray*}
S=\mathbb E\left[Z^\top \hat P\nabla^2 F(\hat x)(y- \hat x)(y- \hat x)^\top \nabla^2 F(\hat x) \hat P Z \right],
\end{eqnarray*}
which is finite, since
\begin{eqnarray*}
   & & \Ebb\left[\norm{Z^\top \hat P\nabla^2 F(\hat x)(y-\hat x)(y-\hat x)^\top \nabla^2 F(\hat x)\hat PZ}\right] \\
  &\leq & \Ebb\left[\norm{Z^\top} \norm{ \hat P\nabla^2 F(\hat x)(y-\hat x)(y-\hat x)^\top \nabla^2 F(\hat x)\hat P}\norm{Z}\right] \\
  &=& \Ebb\left[\norm{Z}^2 \norm{ \hat P\nabla^2 F(\hat x)(y-\hat x)(y-\hat x)^\top \nabla^2 F(\hat x)\hat P}\right] \\
  &\leq & \Ebb\left[\norm{Z}^2\right] \sup_{x\in[0,1]^{|\mathcal{L}|},y\in \{0,1\}^{|\mathcal{L}|}}\norm{\hat P(x)\nabla^2 F(x)(y-x)(y-x)^\top\nabla^2 F(x)\hat P(x)} <\infty
\end{eqnarray*}
where $\norm{Z^\top}=\norm{Z}$ for the Euclidean norm, and the supremum is finite by continuity over a compact set.
Assumption \ref{as:regularity}(d) gives $\Ebb[\norm{Z}^2]<\infty$.
Therefore $\sqrt{N}\nabla Q_N(\beta_0)\Darrow \mathbf{N}(0,H^{-1}SH^{-1})$ by the Central Limit Theorem.

\textit{Condition iv)} By the uniform law of large numbers, $\sup_{\beta\in \mathcal{B}}\norm{\nabla^2 Q_N(\beta)- \nabla^2 Q(\beta)}\parrow 0$.
To see that $\nabla^2 Q(\beta_0)$ is nonsingular, write
\[
h(\beta) := \beta - \Phi(\beta), \qquad
Q(\beta) := \tfrac12 \|h(\beta)\|^2.
\]
Then
\[
\nabla Q(\beta) = \nabla h(\beta)^\top h(\beta),
\]
so
\[
\nabla^2 Q(\beta) = \nabla h(\beta)^\top \nabla h(\beta)
  + \sum_{j=1}^K h_j(\beta)\,\nabla^2 h_j(\beta).
\]
At $\beta_0$ we have $h(\beta_0)=0$ because $\Phi(\beta_0)=\beta_0$, so the second term drops out and
\[
\nabla^2 Q(\beta_0) = \nabla h(\beta_0)^\top \nabla h(\beta_0).
\]
Moreover, $\nabla h(\beta) = I - \nabla\Phi(\beta)$, hence by Lemma~\ref{lem:Phi_properties}
$\nabla h(\beta_0) = I - \nabla\Phi(\beta_0) = I$, and therefore
\[
\nabla^2 Q(\beta_0) = I,
\]
which is symmetric and positive definite.
Then $\sqrt{N}(\bar \beta-\beta_0)\Darrow \mathbf{N}(0,H^{-1}SH^{-1})$ by \citet[][thm. 3.1]{newey_large_1994}.

Finally, we show that $\bar \beta=\tilde \beta$ with probability approaching 1. Since $\beta_0$ is interior to $\mathcal{B}_0$ and $\bar \beta\parrow \beta_0$, $\bar \beta$ is an interior solution with probability approaching $1$, such that $   \nabla Q_N(\bar \beta)=0$, i.e.
\begin{eqnarray}\label{eq:Q_FOC}
    (I-\nabla \Phi_N(\bar \beta))^\top(\bar \beta-\Phi_N(\bar \beta))=0 ,
\end{eqnarray}
where $I-\nabla \Phi_N(\bar \beta)\parrow I-\nabla \Phi(\beta_0)=I$ by Lemma \ref{lem:Phi_properties} which is nonsingular, such that $I-\nabla \Phi_N(\bar \beta)$ is nonsingular with probability approaching $1$. Then \eqref{eq:Q_FOC} implies that $\bar \beta=\Phi_N(\bar \beta)$. Since the minimizer of $Q_N$ is unique with probability approaching $1$, $\bar \beta$ is the only fixed point of $\Phi_N$ in $\mathcal{B}$.

As $\bar \beta$ is a fixed point of $\Phi_N$  with probability approaching one, we have by definition of $\tilde \beta$ in \eqref{eq:tilde_beta} that $
    RV_N(\tilde \beta)\leq RV_N(\bar \beta)$. 
As is standard in extremum estimation, we take the minimization in \eqref{eq:tilde_beta} to run over a compact parameter space $\Theta$ with $\beta_0\in\mathrm{int}(\Theta)$ and $\mathcal{B}\subset\Theta$. By the uniform law of large numbers, $\sup_{\beta\in\Theta}\| RV_N(\beta)-RV(\beta)\|\parrow 0$. Applying this at $\tilde\beta$ and at $\bar\beta$, and using $\bar\beta\parrow\beta_0$ with continuity of $RV$,
\[
RV(\tilde\beta)\;\leq\; RV_N(\tilde\beta)+o_p(1)\;\leq\; RV_N(\bar\beta)+o_p(1)\;=\;RV(\beta_0)+o_p(1).
\]
The set $\Theta\setminus\mathrm{int}(\mathcal{B})$ is compact, $RV$ is continuous on it, and the unique minimizer $\beta_0$ of $RV$ (Lemma~\ref{lem:unique_RV}) lies in $\mathrm{int}(\mathcal{B})$ by construction, so
\[
\delta\;\coloneqq\;\inf_{\beta\in\Theta\setminus\mathrm{int}(\mathcal{B})}RV(\beta)-RV(\beta_0)\;>\;0 .
\]
If $\tilde\beta\notin\mathcal{B}$, then $RV(\tilde\beta)\geq RV(\beta_0)+\delta$ by the definition of $\delta$. Therefore
\[
\Pr\bigl[\tilde\beta\notin\mathcal{B}\bigr]\;\leq\;\Pr\bigl[RV(\tilde\beta)-RV(\beta_0)\geq \delta\bigr]\;\rightarrow\;0,
\]
where the convergence holds because $RV(\tilde\beta)-RV(\beta_0)$ is bounded by an $o_p(1)$ term while $\delta>0$ is a fixed constant. Hence the fixed point $\tilde\beta$ lies in $\mathcal{B}$ with probability approaching $1$, where $\bar\beta$ is the unique fixed point of $\Phi_N$, so
$\tilde \beta=\bar \beta$ with probability approaching $1$.

It now follows by the triangle inequality that
\begin{eqnarray*}
    \norm{\tilde \beta-\beta_0}\leq \norm{\tilde \beta-\bar \beta}+\norm{\bar \beta-\beta_0}\parrow 0.
\end{eqnarray*}
\end{proof}

\section{Additional simulation results}
\label{app:sim_converged}

We provide additional finite-sample results on the converged replications in
Tables~\ref{tab:sim-bias-rmse-conv} and~\ref{tab:sim-sd-se-conv}. 
Conditioning on convergence makes little change to the statistics, compared to the full-replication results in Tables~\ref{tab:sim-bias-rmse} and~\ref{tab:sim-sd-se}. 
Across every coefficient and sample size, the bias and RMSE move by at most $0.013$, the empirical standard deviations and estimated standard errors by at most $0.009$, and the coverage by less than one percentage point.

\begin{table}[H]
  \centering
  \begin{threeparttable}
  \caption{Bias, RMSE, and $\sqrt{N}$-RMSE by coefficient and sample size, over converged replications.}
  \label{tab:sim-bias-rmse-conv}
  \small
  \setlength\tabcolsep{4pt}
  \begin{tabular}{cccccccccc}
    \toprule
    $N$ & \multicolumn{3}{c}{$\tilde\beta^{\text{length}}$} & \multicolumn{3}{c}{$\tilde\beta^{\text{time}}$} & \multicolumn{3}{c}{$\tilde\beta^{\text{sim}}$} \\
    \cmidrule(lr){2-4} \cmidrule(lr){5-7} \cmidrule(lr){8-10}
     & Bias & RMSE & $\sqrt{N}$-RMSE & Bias & RMSE & $\sqrt{N}$-RMSE & Bias & RMSE & $\sqrt{N}$-RMSE \\
    \midrule
    200 & 0.0273 & 0.1726 & 2.4406 & 0.0202 & 0.1454 & 2.0569 & 0.0099 & 0.4480 & 6.3354 \\
    500 & 0.0122 & 0.0996 & 2.2260 & 0.0050 & 0.0805 & 1.7992 & 0.0041 & 0.2588 & 5.7863 \\
    1,000 & 0.0111 & 0.0673 & 2.1274 & -0.0039 & 0.0559 & 1.7668 & -0.0011 & 0.2009 & 6.3528 \\
    2,000 & 0.0009 & 0.0444 & 1.9860 & 0.0031 & 0.0381 & 1.7034 & 0.0049 & 0.1298 & 5.8047 \\
    5,000 & 0.0011 & 0.0277 & 1.9574 & 0.0010 & 0.0229 & 1.6216 & -0.0011 & 0.0806 & 5.7016 \\
    10,000 & 0.0009 & 0.0178 & 1.7756 & -0.0000 & 0.0169 & 1.6853 & 0.0010 & 0.0573 & 5.7289 \\
    20,000 & -0.0003 & 0.0137 & 1.9315 & 0.0003 & 0.0115 & 1.6255 & 0.0007 & 0.0401 & 5.6686 \\
    \bottomrule
  \end{tabular}
  \begin{tablenotes}
      \footnotesize
      \item Computed on the converged replications only: $R_{\mathrm{conv}} = 375, 401, 442, 462, 472, 489, 496$ of $R=500$ for $N = 200, 500, 1,000, 2,000, 5,000, 10,000, 20,000$, respectively. 
    \end{tablenotes}
  \end{threeparttable}
\end{table}

\begin{table}[H]
  \centering
  \begin{threeparttable}
  \caption{Empirical SD, $\overline{\text{SE}}$, and 95\% coverage by coefficient and sample size, over converged replications.}
  \label{tab:sim-sd-se-conv}
  \small
  \setlength\tabcolsep{4pt}
  \begin{tabular}{cccccccccc}
    \toprule
    $N$ & \multicolumn{3}{c}{$\tilde\beta^{\text{length}}$} & \multicolumn{3}{c}{$\tilde\beta^{\text{time}}$} & \multicolumn{3}{c}{$\tilde\beta^{\text{sim}}$} \\
    \cmidrule(lr){2-4} \cmidrule(lr){5-7} \cmidrule(lr){8-10}
     & SD & $\overline{\text{SE}}$ & Coverage & SD & $\overline{\text{SE}}$ & Coverage & SD & $\overline{\text{SE}}$ & Coverage \\
    \midrule
    200 & 0.1706 & 0.1339 & $92.00\%$ & 0.1442 & 0.1184 & $89.87\%$ & 0.4485 & 0.4086 & $93.60\%$ \\
    500 & 0.0989 & 0.0847 & $89.28\%$ & 0.0804 & 0.0746 & $93.77\%$ & 0.2591 & 0.2558 & $93.77\%$ \\
    1,000 & 0.0664 & 0.0599 & $92.76\%$ & 0.0558 & 0.0523 & $93.67\%$ & 0.2011 & 0.1797 & $91.18\%$ \\
    2,000 & 0.0444 & 0.0419 & $94.37\%$ & 0.0380 & 0.0370 & $94.16\%$ & 0.1298 & 0.1268 & $94.81\%$ \\
    5,000 & 0.0277 & 0.0268 & $93.43\%$ & 0.0229 & 0.0233 & $95.97\%$ & 0.0807 & 0.0803 & $94.07\%$ \\
    10,000 & 0.0178 & 0.0189 & $96.93\%$ & 0.0169 & 0.0164 & $93.86\%$ & 0.0573 & 0.0568 & $94.27\%$ \\
    20,000 & 0.0137 & 0.0134 & $93.75\%$ & 0.0115 & 0.0116 & $94.56\%$ & 0.0401 & 0.0402 & $95.36\%$ \\
    \bottomrule
  \end{tabular}
  \begin{tablenotes}
      \footnotesize
      \item Computed on the converged replications only: $R_{\mathrm{conv}} = 375, 401, 442, 462, 472, 489, 496$ of $R=500$ for $N = 200, 500, 1,000, 2,000, 5,000, 10,000, 20,000$, respectively.
    \end{tablenotes}
  \end{threeparttable}
\end{table}

\begin{table}[H]
    \centering
    \caption{Scalability to the number of parameters}
    \label{tab:scalability_K}
    \begin{tabular}{ccccc}
    \toprule
${K}$ & Median runtime [s] & Min runtime [s] & Max runtime [s]& Iterations \\
\midrule
2          & 1.546           & 1.239          & 1.624           & 4              \\
3          & 1.452           & 1.422          & 1.466          & 4              \\
4          & 1.639           & 1.337            & 1.711          & 5              \\
5          & 1.356           & 1.339          & 1.623          & 4              \\
8          & 1.375           & 1.345          & 1.643          & 4 \\
\bottomrule
\end{tabular}
\end{table}

\section{Additional implementation details}
\label{app:implementation_details}

This appendix collects implementation details for the simulation study in Section~\ref{sec:sim}, the scalability study in Section~\ref{sec:scalability}, and the empirical application in Section~\ref{sec:empirical_application}.
All runs are performed on the same workstation, equipped with Intel Xeon Gold 6426Y CPUs (64 cores) and 512~GB memory. 

The quadratic PURC problem (which appears in both the generation of simulated data and in estimation) is essentially a quadratic program. Hence, we exploit the highly optimized quadratic problem solver (PIQP, proximal interior point solver, \cite{schwan2023piqp} with warm starting~\citep{chen2025warmstarting}) to solve the quadratic PURC problem.
We use the same solver parameters throughout: absolute and duality-gap tolerances are set to $10^{-12}$, and we cap the number of interior-point iterations at 1000 (no binding caps were encountered in all experiments).

In the estimator, a trip's active set is determined by the dual-slack sign test as detailed in the paragraph on active link detection at the end of Section~\ref{sec:optimizing_behavior}, with node potentials returned by the PIQP solver.
To improve numerical stability when solving for the projection matrix $\tilde P_n^*$ in Eq.~\eqref{eq:Ptilde}, we further ground the active subnetwork separately on each connected component (in exact arithmetic, there would be only one), so components that are disconnected due to floating residuals are accommodated.

For the microPURC estimator and Algorithm~\ref{iterated_algo}, we use the following nonmonotone line search throughout:
\begin{itemize}
  \item Reduction factor $\gamma = 0.5$.
  \item Nonmonotone window of width $3$ for the line-search reference value.
  \item Maximum number of line searches $m_{\max} = 5$.
\end{itemize}

The fixed-point residual tolerance is $\tau = 10^{-4}\sqrt{K}$ in every experiment reported in the paper. The $\sqrt{K}$ factor keeps the threshold on $\lVert\Phi_N(\tilde\beta)-\tilde\beta\rVert$ comparable across parameter dimensions. 



\paragraph{Sensitivity analysis of algorithm parameters.}
To show the impacts of the algorithm parameters on estimation performance, we vary each one at a time around its default on $D=20$ independent datasets of $N=10{,}000$ trips, and record the largest deviation from the estimate obtained at the default. 
As shown in Table~\ref{tab:sim-algo-params}, four of the six parameters, the iteration cap, the nonmonotone window, the reduction factor, the number of line searches, leave $\tilde\beta$ numerically unchanged. 
Only the fixed-point tolerance $\tau$, and the inner QP tolerance slightly influence the estimate, and the largest deviation anywhere is $0.0084$, an order of magnitude below the statistical error $\mathbb{E}\lVert\tilde\beta-\beta_0\rVert\approx0.0557$ at this sample size. 
Among the sweep values of $\tau$, the largest deviation comes from loosening $\tau$ to $10^{-2}$, which simply stops the iteration earlier. Tightening $\tau$ to $10^{-5}$ only increases the number of iterations ($15$ vs $4$ at the default) and has no impact on the estimates. 

\begin{table}[H]
  \centering
  \begin{threeparttable}
  \caption{Sensitivity of the estimate to Algorithm~\ref{iterated_algo}'s tuning parameters ($D=20$ datasets, $N=10{,}000$).}
  \label{tab:sim-algo-params}
  \small
  \begin{tabular}{llrl}
    \toprule
    Parameter & Values swept & $\max\lVert\Delta\tilde\beta\rVert_\infty$ & Conv.\% \\
    \midrule
    $\tau$ & 0.01, 0.003, 0.001, 0.0003, 0.0001, 1e-05 & 0.0067 & 95--100 \\
    Max.\ outer iterations & 50, 100, 200, 500 & 0.0000 & 100 \\
    Nonmonotone window $M$ & 1, 3, 5, 10 & 0.0000 & 100 \\
    Reduction factor & 0.3, 0.5, 0.7 & 0.0000 & 100 \\
    Max.\ line searches $m_{\max}$ & 3, 5, 10 & 0.0000 & 100 \\
    QP tol.\ $\varepsilon_{\mathrm{abs}}$ & 0.0001, 1e-06, 1e-08, 1e-10, 1e-12, 1e-14 & 0.0084 & 100 \\
    \bottomrule
  \end{tabular}
  \end{threeparttable}
\end{table}

\paragraph{Code availability.}
The code implementing the microPURC estimator and Algorithm~\ref{iterated_algo} is publicly available at \url{https://github.com/andyYaoR/micropurc}, together with
scripts that reproduce the simulation, scalability, and robustness results in Sections~\ref{sec:sim} and~\ref{sec:scalability} and the benchmark networks used there. The empirical application in Section~\ref{sec:empirical_application} uses trip records from the Danish road pricing experiment, which are confidential and cannot be shared; the repository contains the estimation code used for that application.

\section{Empirical application details}\label{app:emp_detail_stats}

\paragraph{Trip-level summary.}
Table~\ref{tab:dataset-summary} reports summary statistics for trip distance, trip duration, and the distribution of OD pairs per user.

\begin{table}[!ht]
\centering
\caption{Summary statistics for filtered observations (control periods, baseline specification).}
\label{tab:dataset-summary}
\begin{tabular}{lrrrrrr}
\toprule
 & Mean & Median & P25 & P75 & Min & Max \\
\midrule
Distance (km) & 14.72 & 9.21 & 4.85 & 18.62 & 2.14 & 102.08 \\
Duration (min) & 18.87 & 15.15 & 9.48 & 24.37 & 5.17 & 80.67 \\
OD pairs per user & 31.39 & 20 & 7 & 43 & 1 & 184 \\
\bottomrule
\end{tabular}
\end{table}

Using an income indicator defined relative to the Danish population average, 2{,}430 observations are above the average, and 15{,}589 are at or below the average.

\paragraph{Link attributes.}
The road network contains $98{,}503$ directed links.
For each day, link attributes are stored in 10 time intervals, indexed as follows: (1) 21:00--24:00 and 00:00--05:00, (2) 05:00--06:00, (3) 06:00--07:00, (4) 07:00--08:00, (5) 08:00--09:00, (6) 09:00--15:00, (7) 15:00--16:00, (8) 16:00--17:00, (9) 17:00--18:00, (10) 18:00--21:00, where each observation is associated with link attributes according to its departure time interval.

The main attributes are 1) link length (km), 2) GPS-based link travel time (min), 3) free-flow link travel time (min), and 4) a junction indicator (interacted with an indicator for urban area).
Link length and the interaction $\text{junction}\times\text{urban}$ are time-invariant.
The GPS-based and free-flow travel times are mutually exclusive for each link in each time interval: if sufficient GPS data exist for a link and time interval, the GPS-based travel time is used and the free-flow travel time is set to zero; otherwise, the free-flow travel time is used and the GPS-based travel time is set to zero. The link attribute statistics are summarized in Table~\ref{tab:network-attrs}.

\begin{table}[h]
\centering
\caption{Link attribute summary \\(Time‑varying stats are reported as ranges across the 10 time intervals)}
\label{tab:network-attrs}
\small
\setlength\tabcolsep{4pt}
\begin{tabular}{lrrr}
\toprule
Attribute & Mean (range) & Median (range) & Non-zero (range) \\
\midrule
Link length (km) & 3.079 & 0.350 & 98{,}503 \\
GPS-based link time (min) & 0.881--1.082 & 0.460--0.581 & 79{,}706--85{,}098 \\
Free-flow link time (min) & 13.952--19.240 & 0.658--1.034 & 13{,}405--18{,}797 \\
Link junction $\times$ urban (binary) & 0.503 & 1 & 49{,}506 \\
\bottomrule
\end{tabular}
\end{table}

\paragraph{Prediction quality under the income-heterogeneous specification.}
Figure~\ref{fig:emp-calibration-income} reports the in-sample performance for the income-heterogeneous specification. The predicted flow-weighted expected route totals again align well with the observed totals along the $45^\circ$ line, with $R^2=0.99$ for total length and $R^2=0.91$ for total travel time.

\begin{figure}[H]
  \centering
  \begin{subfigure}{0.48\textwidth}\includegraphics[width=\linewidth]{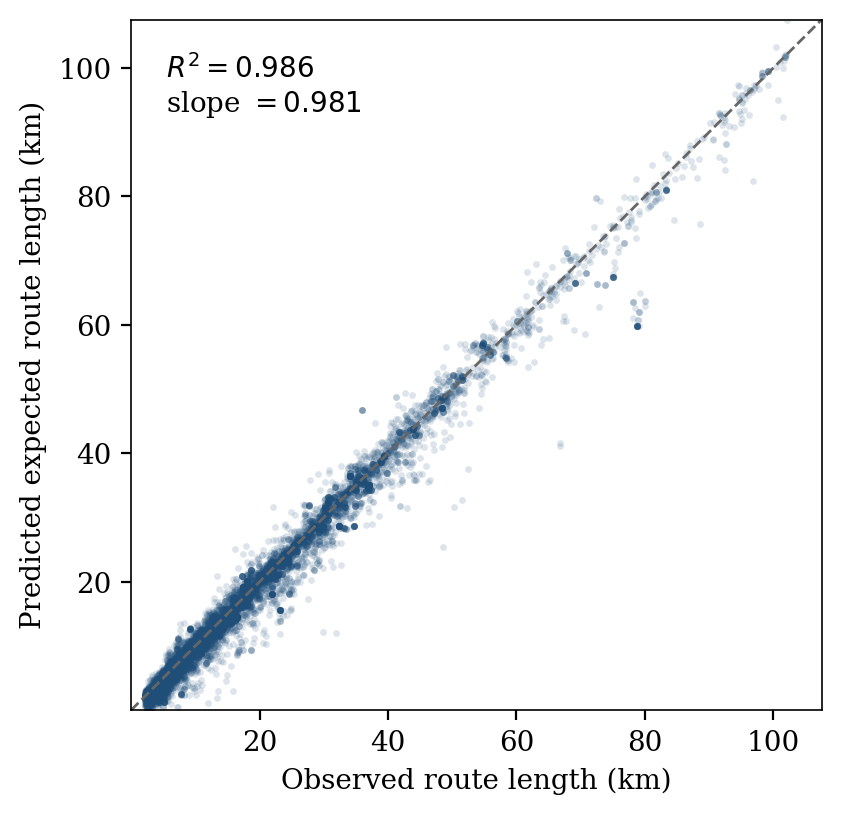}\caption{Route length}\end{subfigure}
  \hfill
  \begin{subfigure}{0.48\textwidth}\includegraphics[width=\linewidth]{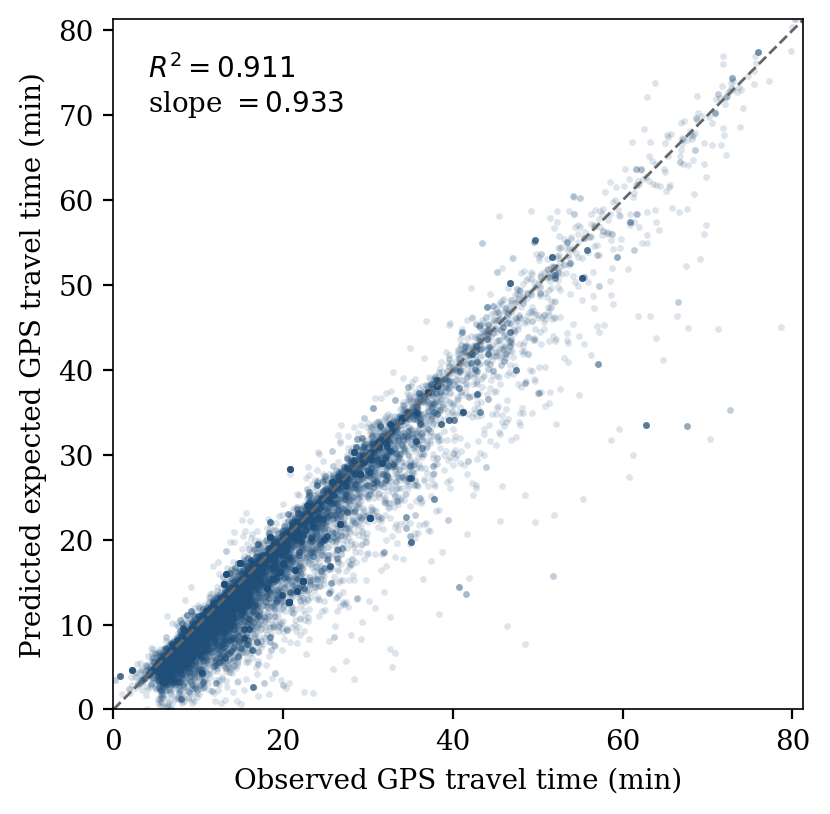}\caption{Travel time}\end{subfigure}
  \caption{Predicted expected versus observed route totals (income-heterogeneous specification, in sample).}
  \label{fig:emp-calibration-income}
\end{figure}

\end{document}